\documentclass[journal]{IEEEtran}
\usepackage{amsmath,amssymb,amsthm,epsfig,subfig,mathrsfs}
\usepackage{xcolor,graphicx,float,comment}
\usepackage[]{algorithm2e}
\usepackage{epstopdf}
\usepackage{lipsum}
\usepackage{cite}
\usepackage{kantlipsum}
\usepackage[font={small}]{caption}
\allowdisplaybreaks

\usepackage{color} 

\newcommand{\Fcal}{\mathcal{F}}
\newcommand{\Acal}{\mathcal{A}}

\newcommand{\Scal}{\mathcal{S}}

\newcommand{\T}{\mathscr{T}}

\newcommand{\Vcal}{\mathcal{V}}

\newcommand{\Ccal}{\mathcal{C}}

\newcommand{\Lcal}{\mathcal{L}}

\newcommand{\Ib}{\mathbf{I}}

\newcommand{\bDelta}{\boldsymbol{\Delta}}
\newcommand{\bXi}{\boldsymbol{\Xi}}
\newcommand{\bUpsilon}{\boldsymbol{\Upsilon}}
\newcommand{\be}{\begin{equation}}
\newcommand{\ee}{\end{equation}}
\theoremstyle{definition}
\newtheorem{remark}{Remark}
\theoremstyle{definition}
\newtheorem{lemma}{Lemma}
\newtheorem{corollary}{Corollary}
\newtheorem{definition}{Definition}
\newtheorem{theorem}{Theorem}
\newtheorem{assumption}{Assumption}
\newtheorem{proposition}{Proposition}

\usepackage{mathtools}
\def\Pre{{\rm Pre}}
\def\Suc{{\rm Suc}}
\def\pre{{\rm pre}}
\def\suc{{\rm suc}}

\DeclareMathOperator*{\argmin}{arg\,min}

\providecommand{\keywords}[1]{\textbf{\textit{Index terms---}} \textbf{#1}}
\begin{document}
\newcounter{mytempeqncnt}
\title{Convergence Results on Pulse Coupled Oscillator Protocols in Locally Connected Networks}
\author{Lorenzo Ferrari, Anna Scaglione, Reinhard Gentz and Yao-Win Peter Hong
\IEEEcompsocitemizethanks{\IEEEcompsocthanksitem L. Ferrari, A. Scaglione and R. Gentz are with the School of Electrical, Computer and Energy Engineering, Arizona State University, Tempe,
AZ, 85281.\protect\\
\IEEEcompsocthanksitem Y.-W. P. Hong is with the Institute of Communications Engineering, National Tsing Hua University, Hsinchu, Taiwan, 30013\protect\\
}
\thanks{}}

\maketitle

\begin{abstract}
This work provides new insights on the convergence of a locally connected network of pulse coupled oscillator (PCOs) (i.e., a bio-inspired model for communication networks)
to synchronous and desynchronous states, and their implication in terms of the decentralized synchronization and scheduling in communication networks.
Bio-inspired techniques have been advocated by many as fault-tolerant and scalable alternatives to
produce self-organization in communication networks. The PCO dynamics in particular have been the source of inspiration for many network synchronization and scheduling protocols. However, their convergence properties, especially in locally connected networks, have not been fully understood, prohibiting
the migration
into mainstream standards. This work provides further results on the convergence of PCOs in locally connected networks and the achievable convergence accuracy under propagation delays. For synchronization, almost sure convergence is proved for $3$ nodes and accuracy results are obtained for general locally connected networks whereas, for scheduling (or desynchronization), results are derived for locally connected networks with mild conditions on the overlapping set of maximal cliques.
These issues have not been fully addressed before in the literature.
\end{abstract}

\keywords{
pulse coupled oscillator, locally connected networks, synchronization, desynchronization, scheduling.}

\IEEEpeerreviewmaketitle

\section{Introduction}

In 1975 Charles Peskin introduced the pulse coupled oscillator (PCO) model to explain the synchronization of pacemaker cells in heart tissues \cite{peskin}. Prior to that,  swarm synchronization among pulsing agents, such as pacemaker cells, was observed frequently in nature \cite{Buck_firefly} but could not be well explained mathematically.
Fifteen years later Mirollo and Strogatz in \cite{strogatz}  proved that fully connected networks of PCOs with excitatory coupling and convex dynamics always converge to fire at unison, except for a measure-zero set of initial conditions.  They also exhamined the case of inhibitory coupling, in which the oscillators emergent behavior turns into a uniformly spaced daisy-chain of pulsing activities among the agents, which can be viewed as a conflict-free schedule of the pulsing activities \cite{nagpal}.
In the early 2000s several groups recognized the applicability of PCO models for network synchronization  \cite{hong2005scalable,Campbell,Mathar,Timme,Motter,Izhikevich,Frigui,lee2008synchronizing,Barbarossa_Celano,Torikai, Nakano} as well as scheduling \cite{nagpal,round_robin_pagliari}, albeit less directly for the latter.
Typically, protocols using these models have a fairly simple signaling mechanism that couple the dynamics of the nodes transmission activities. In turn, these protocols help integrate
the physical and the medium access control layers with network synchronization (typically application layer) activities. Unfortunately, to this day, the impact of network connectivity on the convergence and on the accuracy of the convergent state compared to the desired synchronous state or schedule remains not fully understood.

While the convergence of PCOs to the synchronous state has been studied extensively in the literature, e.g., \cite{strogatz,Kuramoto199115,PhysRevLett.74.1570,PhysRevE.54.5522,2015arXiv150104115O}, little is known for the convergence in locally connected networks\cite{sync_cycle_Nunez,sparsely_connected_sync_Rothkegel,arxiv_review_17}, especially when propagation delays come into play.
The problem of establishing almost sure convergence for locally connected networks remains open, and has only been partially addressed in recent works by imposing additional assumptions on the update dynamics and the initial conditions of the oscillators' phases (see e.g \cite{arxiv_paper_review} which extends the analysis in \cite{arxiv_review_17,arxiv_review_32} for Phase Response Curves (PRC) maps of the {\it delay-advanced type} \cite{arxiv_review_31} and references therein).
In \cite{PCO_accuracy_delays}, a claim on the convergence of the synchronization for a line network was provided, but was only verified through numerical simulations.
Other works, such as \cite{lucarelli,werner2005firefly,degesys2008towards,nagpal,nagpal2}, looked at the asymptotic behavior considering very small coupling between oscillators and focused on the effect of different functions modeling the dynamics of the oscillators, which can be approximated by a continuous-time Kuramoto model \cite{Kuramoto199115}.
These models do not apply for scheduling algorithms \cite{nagpal,round_robin_pagliari} that in all regimes are known to not produce the desired emergent behavior, unless the network is fully connected.
In this work, we revisit the analysis of PCO based synchronization and scheduling over locally connected networks. Compared to the aforementioned literature, we want to highlight that:
(a) the dynamics we chose to analyze simplify the implementation both through digital as well as analog circuits \cite{apsel1,apsel2,apsel3}; (b) we do not need to assume that nodes can separate the signals fired by different nodes when the firing occurs at unison, which may happen when the nodes are close to synchrony (since the absorption property treats multiple interfering pulses as one) (c) our results  hold irrespective of the initial conditions.
Furthermore, the absorption property allows perfect synchronization to occur after some time $t$ when no delay exists; and, when there are delays,  it allows (thanks to the refractory period) the nodes phase difference to remain fixed and converge to a value that is bounded by the maximum sum of the propagation delays over any path in the network after some time $t$. In most of the competing models mentioned above, convergence occurs only asymptotically, as time goes to infinity. 
 Our main contributions are as follows: 1) we show that, for any positive coupling strength, the synchronous state is the unique fixed point for our model (see Section \ref{sec:PCO_conv}) and for a $3$ nodes locally-connected network, convergence to the synchronous state occurs almost surely; 2) we study the effect of propagation delays, and extrapolate the synchronization accuracy expected for more complex topologies and random delays (confirmed by numerical simulations in Section \ref{simulation_results}) by characterizing the set of fixed points for our model in the presence of delays; 3) we adapt the PCO-based scheduling scheme introduced in \cite{round_robin_pagliari} to locally connected networks and analyze its convergence in a general class of these networks, which includes both star and line networks (see Section \ref{sec:PFS_conv}).
The simulations in Section \ref{simulation_results} corroborate our claims. 
This work constitutes the analytical foundation for the design and implementation of the PCO-based distributed scheduling protocol presented in \cite{PulseSS_description} for multi-hop networks, where a comparison with other solutions is also presented. The key-strength of the protocol is that it offers a decentralized solution for two major problems in network communications: a) clock network distribution and b) channel resource allocation. It is then natural to study the attainable clock distribution accuracy when propagation delays come into play (see Section \ref{sec:PCO_delays}) and the equivalent capacity available to each node in the network for our scheduling mechanism. In Remark \ref{remark:coloring} we also discuss the connection between our algorithm's achievable schedule and the solution of the minimum coloring graph problem.
\subsection{Notation}
Unless specified otherwise, we represent matrices and vectors with capital bold-face letters (i.e., $\mathbf{B}$), we use calligraphic letters to indicate sets of nodes (i.e., $\mathcal{A}$), and the greek letters $\Phi$ and $\Psi$ for the internal clocks of the nodes. The letter $\pi$ is used to refer to an index after a permutation and the lowercase letters $\alpha,\beta,\delta$ are positive parameters of the coupling equations. When not used as a node index, $j=\sqrt{-1}$. The notation $t^+$ refers to an instant   after time $t$ when an event triggered a clock update.
We use suffixes to refer to a specific node, pair of nodes or a specific clique in Section \ref{sec:PFS_conv}.
Note also that in our treatment the time $t$ and the propagation delays are normalized by the PCO period ($T_{PCO}$) in Section \ref{sec:PCO_conv} and by the frame duration ($\T$) in Section \ref{sec:PFS_conv}.
\section{Pulse Coupled Synchronization Convergence}\label{sec:PCO_conv}
\begin{figure}[t]
\centering
\includegraphics[scale=0.4]{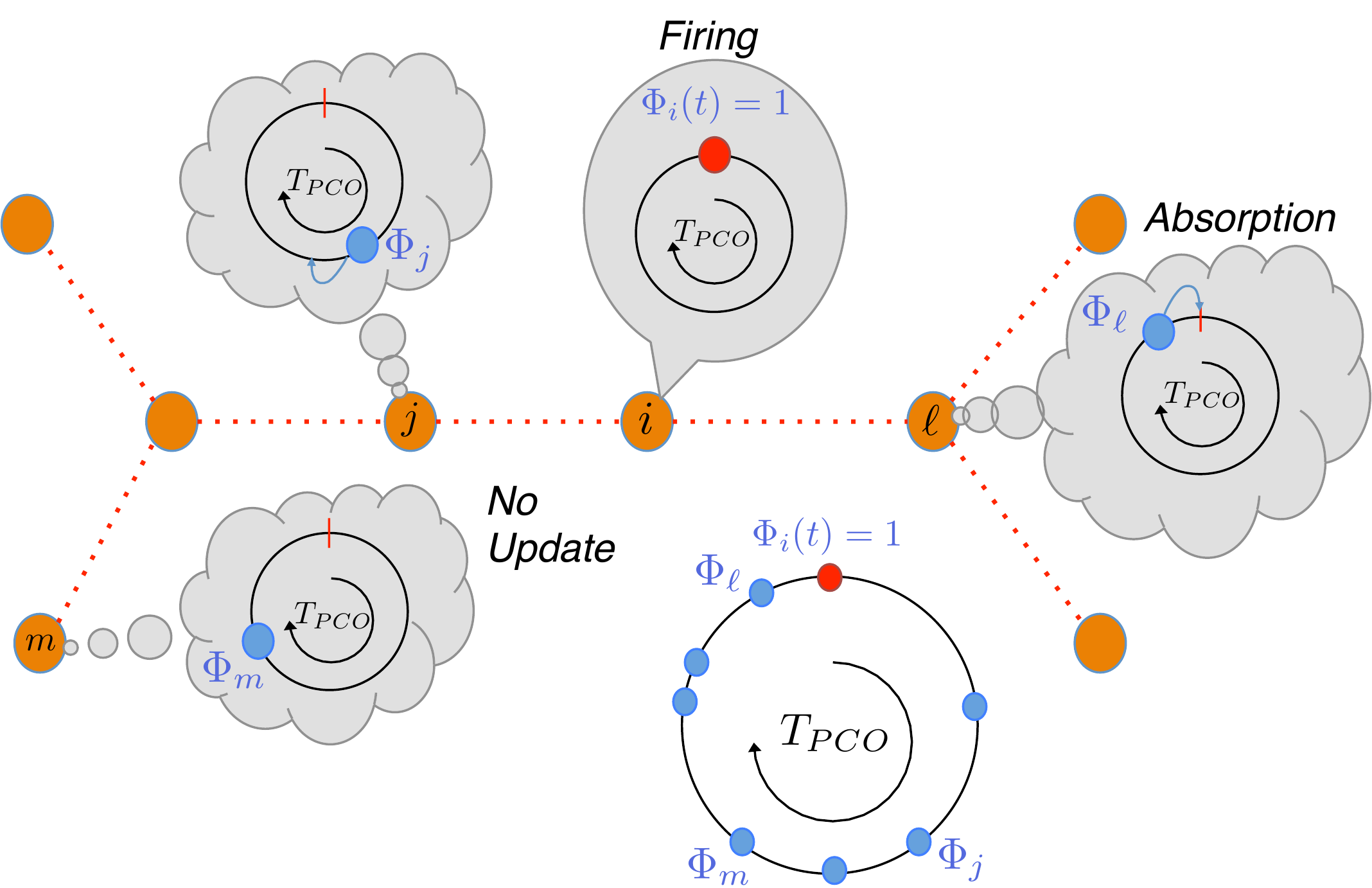}
\caption{Example of a locally connected network of PCOs
at the time node $i$ fires, triggering the phase update of nodes $\ell$ and $j$.
The location of the ball relative to the red mark on the top of each circle indicates the phase of each node, and a ball reaching the red mark indicates the occurrence of a firing.}\label{fig:PCO_initial_state}
\end{figure}
All PCO based algorithms rely on two common features: 1) the emission of beacon signals (or pulses) by each agent in the network and, 2) on the agents updates of their local timers (i.e., PCO clocks) upon reception of beacon signals from their neighbors. The emission of a beacon signal is referred to as the event of firing. An agent fires each time its local timer expires and, in this way, triggers its neighbors to adjust their local PCO clocks ahead, reducing the time until their next firing. 
The preamble signals commonly defined at the physical layer of communication systems can be used as the firing signals of PCO based algorithms, without additional overhead at any layer to support these protocols.

The timer at node $i$ can be modeled by the phase variable:
\begin{equation}
\Phi_i(t)= \frac{t}{T_{PCO}}+\phi_i \mod 1,
\end{equation}
where $\phi_i$ is the initial phase of the timer. When in isolation, each timer increases from $0$ to $1$ repeatedly with period $T_{PCO}$, which is normalized to $1$ without loss of generality.
When placed within the transmission range of each other, the firing of each node will trigger a phase update at any node that receives the firing.
In Peskin's leaky integrate-and-fire model, upon hearing node $i$'s firing at time $t_i$, node $j$ updates its local auxiliary state variable $X_j(t_i)=g(\Phi_j(t_i))$ by the amount $\epsilon$. The inverse mapping of the updated state variable then leads to a jump in the phase of the timer as follows:
\begin{equation}\label{eq:pco_general}
\Phi_j(t_i^+)=\min\left\{g^{-1}(X_j(t_i)+\epsilon),1\right\},
\end{equation}
where the constant $\epsilon$ is called the {\it coupling strength} and $t_i^+$ represents the time immediately following $t_i$.
The function $g$ is called the {\it PCO dynamic} and governs the behavior of the PCO network. It has been shown in \cite{strogatz} that, if $g$ is smooth, monotonically increasing, and concave down, then synchronization in a fully connected network of PCOs is guaranteed to occur, except for a set of initial conditions with measure zero.  More specifically, by choosing $g$ such that $g(x)=\log x$ and $\epsilon=\log(1+\alpha)$, \eqref{eq:pco_general}  equals \cite{Buck_firefly}:
\begin{equation}\label{eq:pco}
\Phi_j(t_i^+) = \min\left\{(1+\alpha)\Phi_j(t_i),1\right\}
\end{equation}
where $\alpha>0$ is the {\it excitatory} coupling factor.
Such choice for $g(x)$ is motivated by the convenience in the implementation of \eqref{eq:pco}, while convergence for the fully connected network is guaranteed in \cite{strogatz}. 
If the phase of node $j$ falls between $\frac{1}{1+\alpha}$ and $1$ at the time of firing by node $i$ (i.e., if $\Phi_j(t_i)\in(\frac{1}{1+\alpha},1]$), then the phase of node $j$ will become $1$ upon detection of the firing event of node $i$ (i.e., $\Phi_j(t_i^+)=1$) and will be triggered to fire immediately as well. The event is called the {\it absorption} of node $j$ by node $i$. In a fully connected network, the absorption between two nodes remains permanent and will continue to occur progressively between clusters of nodes until synchrony is attained.
In the next subsection we analyze the PCO synchronization with local connectivity of the network neglecting the propagation delays that will be included in \ref{sec:PCO_delays}.
\subsection{PCO Synchronization with Local Connectivity}\label{sec:PCO_connectivity}
Let $\cal G=(\cal V,\cal E)$ be an undirected graph that represents the network topology, not necessarily fully connected, and let $e_{ij}=1$, if $ij \in \cal E$, and $e_{ij}=0$ otherwise. We set $e_{ii}=0$ $\forall i \in \cal V$.
The update equation in \eqref{eq:pco} can be modified as
\begin{equation}\label{eq:update_local_connectivity}
\Phi_j(t_i^+) = \min\left\{(1+\alpha e_{ij})\Phi_j(t_i),1\right\}
\end{equation}
Let us define the vector $\boldsymbol{\Delta}(t)$ with entries:
\begin{equation}\label{eq:delta_1_definition}
\Delta_{ij}(t)=\min\{(\Phi_i(t)-\Phi_j(t))\!\!\!\!\!\!\mod{1},(\Phi_j(t)-\Phi_i(t))\!\!\!\!\!\!\mod{1}\}
\end{equation} 
for all $ij$ such that $e_{ij}=1$ (by definition we have $\Delta_{ij}(t)=\Delta_{ji}(t),~\forall t$). 
We then introduce the following: 
\begin{definition}\label{def:sync}
{\it 
A network $\cal G=(\cal V,\cal E)$ of PCOs is said to reach a fixed point at time $t^*$ if $\forall t>t^*$ we have $\bDelta(t)=\bDelta(t^*)$. 
If, in addition, we have that $\bDelta(t^*)=\mathbf{0}$ we say the network is synchronized (or has reached the synchronous state).
}
\end{definition}
We can now show: 
\begin{proposition}\label{PCO_no_delays_fixed_point}
{\it For a locally connected network of PCOs that follow the dynamics in \eqref{eq:update_local_connectivity} with $\alpha>0$, the synchronous state is the unique fixed point (as per Definition \ref{def:sync}) i.e.,  
\be \forall t>t^*, \bDelta(t)=\bDelta(t^*)\Leftrightarrow \bDelta(t^*)=\mathbf{0}\ee.}
\end{proposition}
\vspace{-.2cm}
The proof is in Appendix \ref{proof_PCO_no_delays_fixed_point}. 
We also have the following proposition, proven in Appendix \ref{proof_PCO_conv_3_nodes}:
\begin{proposition}\label{PCO_conv_3_nodes}
{\it On any connected network of $|\mathcal{V}|=3$ PCOs following the dynamics in \eqref{eq:pco} with $\alpha>0$, convergence to the synchronous state ($\bDelta=\mathbf{0}$) occurs almost surely from any initial condition. In addition, if we attach a node with a random initial phase to a synchronized network with an arbitrary topology, the overall network will convergence to the synchronous state almost surely.}
\end{proposition}
The second part of our proposition provides the practical insight that if protocol that allows nodes to join one by one during system setup starting from an arbitrary group of three nodes, then almost sure convergence is guaranteed.
Note that we define almost sure convergence for the case $|{\cal V}|=3$ as done in \cite{strogatz}, where we have convergence to synchronization (i.e. to the fixed point $\bDelta=0$) except for a measure zero set of initial conditions. The proof focuses on the case not covered by \cite{strogatz}, of the line network with nodes $\{1,2,3\}$ and edges $\{(1,2),(1,3)\}$.
\subsection{PCO Synchronization with Local Connectivity and Delays}\label{sec:PCO_delays}
To account for the propagation delays, which include the signal duration, travel time, processing time etc., we define:
\begin{equation}
r_{ij}=t_i+\tau_{ij}
\end{equation}
where $t_i$ is the time node $i$ fires and $\tau_{ij}$ is the delay (expressed in time units equal to the $T_{PCO}=1$). If node $j$ is not in node $i$'s neighbourhood (i.e. $e_{ij}=0$), $\tau_{ij}=0$ and $r_{ij}=t_i$, otherwise $r_{ij}$ represents the time node $j$ is aware of node $i$'s firing.
We assume that all propagation delays are shorter than the PCO period, i.e. $\forall ij \in \mathcal{E}, \tau_{ij}<1$ and that the delays are symmetric, i.e. $\tau_{ij}=\tau_{ji}$. 
In the presence of these propagation delays, PCO protocols cannot converge unless they include a {\it refractory period} \cite{peskin}, i.e. a portion of the cycle, right after their firing event, during which the node does not update its phase.
Notice that, in the absence of the refractory period, the firing of a node may trigger the neighboring node to fire right after the propagation delay, causing the node that originally fired to update after a roundtrip propagation delay from its initial firing event. Let $\rho$ be the duration of the refractory period; this is the so-called {\it echo} effect which can be avoided if: 
\begin{equation}\label{rho}
\rho>2\max \tau_{ij}.
\end{equation}
The update equation in the presence of delays is:
\begin{equation}\label{delay}
\Phi_j(r_{ij}^+)\!
=\!\!\left\{\!\!\!
\begin{array}{l l}
\min\{(1+\alpha e_{ij})\Phi_j(r_{ij}),1\},& \rho
<\!\Phi_j(r_{ij})\!\!\!\!\mod{1} \\
\Phi_j(r_{ij})\!\!\!\!\mod{1},& \mbox{else}.
\end{array}
\right.
\end{equation}
Note that, if $e_{ij}=1$:
\begin{equation}\label{eq:delayed_phase}
\Phi_j(r_{ij})=\Phi_j(t_{i})+\tau_{ij}\pmod{1}
\end{equation}
is the value of the clock phase of node $j$ at the time it detects node $i$'s firing. In this case, $\tau_{ij}$ can be viewed as an additive timing error and the update can be written as follows:
\begin{equation}\label{delay2}
\Phi_j(r_{ij}^+)\!=\min\{(1+\alpha\hat{e}_{ij}(t_i))(|\Phi_j(t_i)+\tau_{ij}|\mod{1}),1\}
\end{equation}
where $\hat{e}_{ij}(t_i)$ is defined as:
\begin{equation}\label{indicator}
\hat{e}_{ij}(t_i)=\begin{cases} 1 ~ &\mbox{if}~ e_{ij}=1~\mbox{and}~|\Phi_j(t_i)+\tau_{ij}|\mod 1>\rho\\
0 &\mbox{else}
\end{cases}
\end{equation}
and can be seen as the element of a time varying adjacency matrix.
We then can prove the following:

\begin{proposition}\label{fixed_points_theorem_delays}
{\it
For deterministic $\tau_{ij}<+\infty~\forall i,j$, if we include a refractory period $2\max \tau_{ij}\leq\rho<\frac{1}{2}+\min\tau_{ij}$, we have that for any locally connected network of PCOs following the dynamics in \eqref{eq:update_local_connectivity}: 
\be 
\forall t>t^*, \bDelta(t)=\bDelta(t^*)\Leftrightarrow \bDelta(t^*)\in\Fcal
\ee 
where 
\be 
\Fcal\triangleq\{\bDelta: 0\leq \Delta_{ij}\leq \tau_{ij}~~~\forall i,j~~\mbox{s.t.}~~ e_{ij}=1\}
\ee 
represents the set of possible fixed points for the algorithm.
}
\end{proposition}

The proof can be found in Appendix \ref{proof_theorem_delays} where the upper-bound for $\rho$ is also discussed. A direct consequence of this proposition is that in order to have a possible choice for the refractory period $\rho$ we need  $\max_{ij}\tau_{ij}<\frac{1}{2}+\min_{ij}\tau_{ij}$
The convergence of the protocol is represented in Fig. \ref{fig:PCO_convergence_state} for the same topology as in Fig. \ref{fig:PCO_initial_state}.
\begin{figure}
\centering
\includegraphics[scale=0.4]{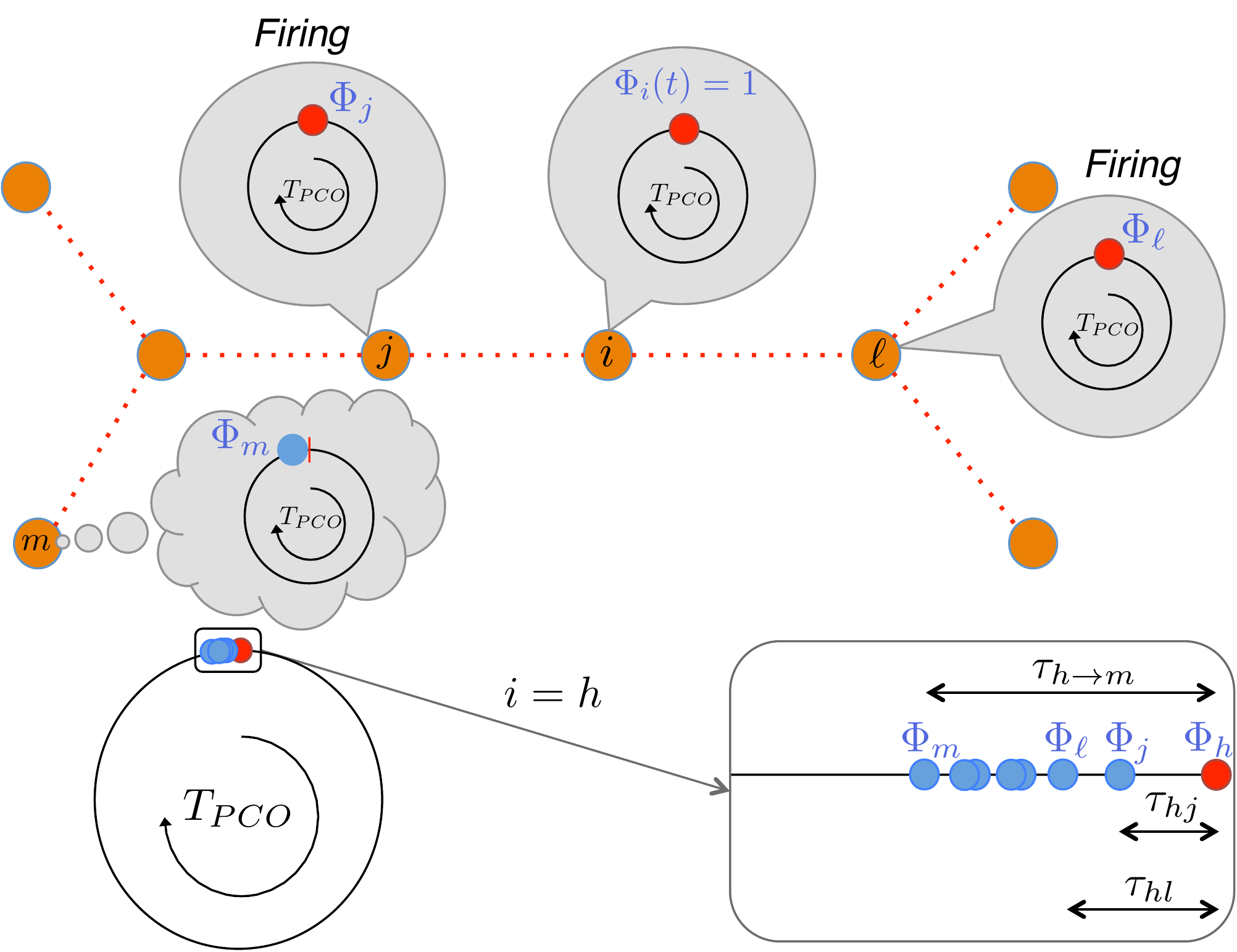}
\caption{Convergence of the PCO protocol as stated in Proposition \ref{fixed_points_theorem_delays} for the network topology of Fig.\ref{fig:PCO_initial_state} with $i$ as the \textit{head} node.}\label{fig:PCO_convergence_state}
\end{figure}
Let $\mathcal{P}_{ij}$ be the set of edges forming the shortest path between node $j$ and node $i$. We can define the accumulated propagation delay on the path from $i$ to $j$:
\begin{equation}
\tau_{i\rightarrow j}=\sum_{\ell m \in \mathcal{P}_{ij}}\tau_{\ell m}
\end{equation}
where $\tau_{\ell m}$ is the time that has elapsed between the actual firing by node $\ell$ and the observation of the firing by node $m$. Clearly $\tau_{i\rightarrow j}=\tau_{ij}$ if $e_{i,j}=1$, i.e., if node $j$ can directly hear the firing of node $i$.
Then we notice that, as long as $\max_{ij}\tau_{i\rightarrow j}<\frac{1}{2}$ for any fixed points in $\Fcal$ it is possible to consider a node $h$ we name the {\it 
head} (not necessarily unique) such that  its minimum distance $\Delta_{hj}$ defined in \eqref{eq:delta_1_definition} $\forall j \in\Vcal$ is:
\be\label{eq:property_head_node}
\Delta_{hj}=(\Phi_h-\Phi_j)\mod{1}
\ee 
since the node is ahead of every other node. We indicate this condition by saying the \textit{head} node ``preceeds all the other nodes".
Since we have $\bDelta\in\Fcal$ it is clear we have the following upper-bound:
\begin{equation}\label{eq:phi_error}
\bDelta_{hj}\leq\tau_{h\rightarrow j},
\end{equation}
Our simulation shows that the bound in \eqref{eq:phi_error} is actually tight, and that is due to the fact that when the initial phases are spread around the PCO cycle, absorptions tend to occur in a cascade and each $j$ node that is absorbed by node $i$ remains at distance $\tau_{ij}$.
The reason why Proposition \ref{fixed_points_theorem_delays} has an inequality instead of an equality is  if the initial conditions are such that two nodes are closer than their propagation delay they can remain at that closer distance relative to their propagation delay, due to the presence of the refractory period that makes all these cases fixed points.
The residual synchronization error can be defined as:
\begin{equation}
\Delta_{\max}=\max_{i,j}\Delta_{ij}.
\end{equation}
At this point, as a direct consequence of the bound \eqref{eq:phi_error} and the property of the head-node in \eqref{eq:property_head_node} we can derive the following expression for the expected residual synchronization error:
\begin{equation}\label{residual_sync_error}
\mathbb{E}\{\Delta_{\max}\}\leq\sum_{h=1}^{N}p_h\left(\max_j{\tau_{h\rightarrow j}}\right)
\end{equation}
where $p_h$ indicates the probability that $h$ is the {\it head} node which depends on the topology and the initial conditions. Although a general characterization of $p_h$ is complex, the expression in \eqref{residual_sync_error} allows us to bound the residual synchronization error considering the best and the worst case scenario for the term $\max_j{\tau_{h\rightarrow j}}$, which is immediately derivable from the topology. A note of caution is that the term {\it best case} indicates the smallest possible bound over all possible choices of the head node, and not the best attainable synchronization error, since in principle the synchronous state $\bDelta=\mathbf{0}$ is a fixed point also for the model with delays.
$\mathbb{E}\{\Delta_{\max}\}$ is the metric we consider in the simulations to evaluate the performances of a line and a star network of fixed length and increasing density of nodes.
The expected value $\mathbb{E}\{\Delta_{\max}\}$ is an interesting metric to characterize the protocol performance and may be more insightful compared to the worst and the best case scenario, directly computable through our analysis in this section, because the latter remain identical over a wide variety of networks while $\mathbb{E}\{\Delta_{\max}\}$ changes. 
\begin{remark}
If the propagation delays are random and bounded by $\tau_{max}$, the result continues to apply as long as $\rho>2\tau_{max}$.
The fixed point is still compatible with the last realization of random delays characterizing all absorptions until the last.
\end{remark}
 The direct consequence of Proposition \ref{fixed_points_theorem_delays} is that in multi-hop networks, propagation delays tend to accumulate worsening the overall synchronization accuracy. This limits the application of PCO as a clock distribution mechanism in very large networks. To overcome this problem, we propose in \cite{PulseSS_description} to couple the synchronization and scheduling with the purpose of separating firing events to give each node the possibility to estimate the propagation delays $\tau_{ij}$ and compensate for them in their updates, thus improving the final synchronization accuracy which will be bounded by the cumulative error in these estimates.
\section{Pulse Coupled Scheduling Convergence}\label{sec:PFS_conv}
In addition to achieving synchronization, the update equation in \eqref{eq:pco} can also be modified to trigger a daisy chain of transmissions from the nodes and attain desynchronization. For this purpose, it is necessary to choose $\alpha<0$, which in this case is called \textit{inhibitory} coupling.
In the case of inhibitory coupling, nodes are propelled to increase their phase differences, so absorptions cannot occur. For a fully connected network, it has been shown in \cite{round_robin_pagliari} that the nodes phases converge to a limit cycle in which the time between adjacent firing events is a constant; this state is named {\it weak desynchronization}, since the firing times continue to shift.
For locally connected networks, however using PCO with inhibitory coupling does not lead to global desynchronization (in a weak or strict sense).
In this part, we are then going to explore the behavior of the scheduling algorithm introduced in \cite{round_robin_pagliari} (different from \eqref{eq:pco}) and discuss the convergence for locally connected topologies.
We will consider a generic connected multiclique graph, where a clique might or not have nodes in common with another clique. Note that any graph can be viewed as a multiclique graph.
We are interested in enforcing desynchronization between nodes in the same clique because this will lead to a TDMA schedule, enabling nodes to transmit in different portion of the time frame avoiding, ideally, collisions and loss of data.
Achieving TDMA scheduling in a decentralized fashion is very complex but also extremely appealing for many emerging applications of wireless sensor networks (the so called the Internet of Things (IoT)), where a dramatic growth in the number of sensors makes centralized TDMA solution impractical. 
Our solution is not only able to assign different portion of the frame to the network nodes', but also leaves white spaces in the frame to enable a different asynchronous system, or new nodes joining the network, to find transmission opportunities. 
This could useful for the coexistence of different protocols. 
\subsection{From PCO desynchronization to Proportional Scheduling}
The protocol considered here (and previously proposed in \cite{round_robin_pagliari}) relies on the evolution of two local continuous timers $(\Phi_i(t),\Psi_i(t))$ for each node $i$, indicating the start and the end of its assigned portion of the frame.
The two timers $\Phi_i(t)$ and $\Psi_i(t)$ can be viewed as the phase of two PCOs that evolve from $0$ to $1$ repeatedly with period equal to the frame duration $\T$. Signal messages (i.e., pulse firings) are emitted by node $i$ at the time instants corresponding to $\Phi_i(t)=1$ and $\Psi_i(t)=1$ to indicate the start and end times of its scheduled duration.
In this section we focus on the algebraic description of the algorithm and ignore issues regarding the acknowledgment of signals and the compensation of propagation delays. The practical treatment of these issues can be found in \cite{PulseSS_description}.

A locally connected network consists of multiple complete subgraphs, called cliques. Let $\cal C$ be the set of maximal cliques, i.e., cliques that cannot be made larger by including any additional node. Moreover, let ${\mathcal V}_c\subset{\mathcal V}$ be the set of nodes contained in clique $c\in\cal C$ and, for each $i\in\mathcal{V}$, we define $\mathcal{C}_i\triangleq \{c\in\mathcal{C}: i\in\mathcal{V}_c\}$ as the set of cliques node $i$ belongs to.
A clique can be viewed in practice as an abstraction of nodes in a cluster that are connected through the relaying of signals by their cluster head (CH) \footnote{In wireless networks, a multiclique graph can be obtained by leveraging the presence of CHs.
The CH is connected with a set of nodes and for that set it acts like a bridge, re-broadcasting the firing signals from the single nodes to all the nodes it has in its communication range: in this way, the CH basically emulates an all-to-all connected network among the nodes it communicates with, generating a clique in the graph topology of our network.}.
We denote the set of {\it shared} (or {\it gateway}) {\it nodes} ${\cal S}_{cc'}$, i.e. nodes that belong to the two cliques $c$ and $c'$ (mathematically ${\cal S}_{cc'}={\cal V}_c \cap {\cal V}_{c'}$) and we indicate with $\Lcal_c$ the set of {\it local} nodes for clique $c$, i.e. nodes that only belong to that clique (formally $\Lcal_c=\Vcal_c\setminus\bigcup_{c'\in\Ccal}\Scal_{cc'}$).

Moreover, we assume that the initial values of the timers $\{(\Phi_i(t_0), \Psi_i(t_0))\}_{i\in{\cal V}}$ are chosen such that the collision avoidance condition is satisfied, that is, for any $c \in \mathcal{C}$ and $i,j \in \mathcal{V}_c$:
\begin{equation}\label{collision_avoidance_condition}
\Phi_i(t_0) - \Psi_i(t_0)\leq \Phi_i(t_0) - \Phi_j(t_0)
\end{equation}
where the above operations are modulo $1$.
In this case, we can denote by $\pi_k^c(t)$ the $k$th index at time $t$ of the permutation of the nodes' indices that sorts the phase variables of the nodes in ${\cal V}_c$ in descending order at time $t$ i.e., in the order such that
$$\Phi_{\pi_1^c(t)}(t)>\Phi_{\pi_2^c(t)}(t)>\cdots>\Phi_{\pi^c_{|\mathcal{V}_{c}|}(t)}(t).$$
In the following, we shall omit the time index $t$ in $\Phi_{\pi_k(t)}(t)$ whenever its dependence on $t$ is clear.
For this algorithm, as it will be clear later, the firing order does not change over time.
To simplify the notation, let us consider two functions $\text{pre, suc}: {\mathcal V}\times{\mathcal C}  \rightarrow{\mathcal V}$, defined by:
\begin{align}
\text{pre}(i,c)&=\pi_{k-1}^c~~\in {\mathcal V}_c\label{eq:pre_i_c}\\
\text{suc}(i,c)&=\pi_{k+1}^c~~\in {\mathcal V}_c\label{eq:suc_i_c}
\end{align}
for all $c\in{\cal C}$ and for all $i\in{\cal V}_c$ such that
$i=\pi_k^c$, where $\pi_{0}^c=\pi_{|{\cal V}_c|}^c$ and $\pi_{|{\cal V}_c|+1}^c=\pi_{1}^c$ (the above quantities are not defined if $i \notin \mathcal{V}_{c}$).
Here, pre($i,c$) and suc($i,c$) represent the nodes in ${\cal V}_c$ that produce a firing event (the expiraton of one of the two timers) immediately before and after the firing of the start and the end timers of node $i$ (see Fig.\ref{fig:PFS_initial_state}-\ref{fig:PFS_convergence_state}-\ref{fig:shared_node_update}). Moreover, since node $i$ may belong to more than one clique, it is necessary to define
\begin{align}
\text{Pre}(i,t)&=\text{pre}\left(i,\argmin_{c'\in\mathcal{C}_i}\{\Psi_{\pre(i,c')}(t) - \Phi_i(t) \}\right)\label{eq:Pre_def}\\
\text{Suc}(i,t)&=\text{suc}\left(i,\argmin_{c'\in\mathcal{C}_i}\{  \Psi_i(t) - \Phi_{\suc(i,c')}(t) \}\right)\label{eq:Suc_def}
\end{align}
as the two nodes (not necessarily in the same clique) which transmit immediately before and after node $i$ among all the ones in cliques node $i$ belongs to.
Here, the phase differences are  modulo $1$.
Notice that we need node $i$ to have no conflicts with all the nodes in the cliques that it belongs to, which can be more than one.
The identification of these two nodes is fundamental for the bio-inspired procedure we are going to introduce in the next section.\footnote{To discriminate the firing times of $\Pre({\pi_k^c})$  from the others, node ${\pi_k^c}$ can prepare an update for any firing it hears and discard the update if a more recent firing event is registered.  $\Suc({\pi_k^c})$ is easy to identify since is the first to fire after the expiration of node ${\pi_k^c}$'s own end clock. The value of $\Psi_{\Pre({\pi_k^c})}(t)$, which is the reference for node ${\pi_k^c}$ to make its update, can be calculated simply measuring the time it elapsed between the firing of  $Pre({\pi_k^c})$  and the local clock. Hence all the information needed to advance the protocol is implicitly available and firing beacons do not need to carry data, but rather can be special preambles that are easy to detect at the PHY layer  \cite{PulseSS_description}.\label{footnote3}} 
\subsection{Scheduling procedure}
Different from the coupling mechanism used for synchronization, in which any firing event triggers updates by all nodes in range, in the scheduling algorithm, an update for every node $i$ happens just once every round (i.e. every expiration of the local timers) and only in response to the firing signals of the two nodes, $\text{Pre}(i,t)$ and $\text{Suc}(i,t)$, whose timers expire just before and after node $i$'s own local start and end timers. Here we denote the frame duration by ${\T}$ which is also the time unit with which we measure time. 
In the rest of the article we will omit the dependence on time of $\Pre(i)$ and $\Suc(i)$ since the time instant is always indicated in the equations where they appear (see also footnote \ref{footnote3} in the previous page).
Note that the time at which the firing of $\text{Pre}(i)$ is received is stored by node $i$, but the actual update takes place after $\text{Suc}(i)$ firing event is also detected. We refer to this time instant as $t_i$ (assuming then $\Phi_{\Suc(i)}(t_i)=0$).
To perform the update, node $\pi_k^c$ first evaluates a target value for its two local timers $(\Phi_{\pi_k^c}^*,\Psi_{\pi_k^c}^*)$ as:
\begin{align}
\Phi_{\pi_k^c}^*(t_{\pi_k^c})&=\frac{D_{\pi_k^c}\!+\!\delta}{D_{\pi_k^c}\!+\!2\delta}\Psi_{\Pre({\pi_k^c})}\!(t_{\pi_k^c})\label{eq.Ntarget_s}\\
\Psi_{\pi_k^c}^*(t_{\pi_k^c})&=\frac{\delta}{D_{\pi_k^c}\!+\!2\delta}\Psi_{\Pre({\pi_k^c})}\!(t_{\pi_k^c})\label{eq.Ntarget_e}
\end{align}
where $D_{\pi_k^c}$ is a parameter chosen to reflect the demand of node $\pi_k^c$ in a frame and $\delta$, equal for every node, represents the relative portion of the frame that should be reserved as guard period between transmission of different nodes.
The new phase of these timers are set to be a convex combination of their respective target values and current phases, i.e.,
\begin{align}
\Phi_{\pi_k^c}(t_{\pi_k^c}^+)&=(1-\beta) \Phi_{\pi_k^c}(t_{\pi_k^c}) + \beta \Phi_{\pi_k^c}^*(t_{\pi_k^c})\label{eq:s_update}\\
\Psi_{\pi_k^c}(t_{\pi_k^c}^+)&=(1-\beta) \Psi_{\pi_k^c}(t_{\pi_k^c}) + \beta \Psi_{\pi_k^c}^*(t_{\pi_k^c})
\label{eq:e_update}
\end{align}
with $\beta\in(0,1)$. In light of equations \eqref{eq:s_update}-\eqref{eq:e_update} and the definition of $\Pre(i,t)$ and $\Suc(i,t)$ in \eqref{eq:Pre_def}-\eqref{eq:Suc_def} there is no overlap of local timers in each clique, i.e. the firing order in each clique remains fixed (see also Fig.\ref{fig:shared_node_update} in Appendix \ref{proof_PFS_convergence}).
The update procedure for node $i$ at the time of firing of node $j=\text{Suc}(i)$ is illustrated in Fig.\ref{fig:PFS_initial_state} for a topology with two cliques.
\begin{figure}[t]
\centering
\includegraphics[scale=0.38]{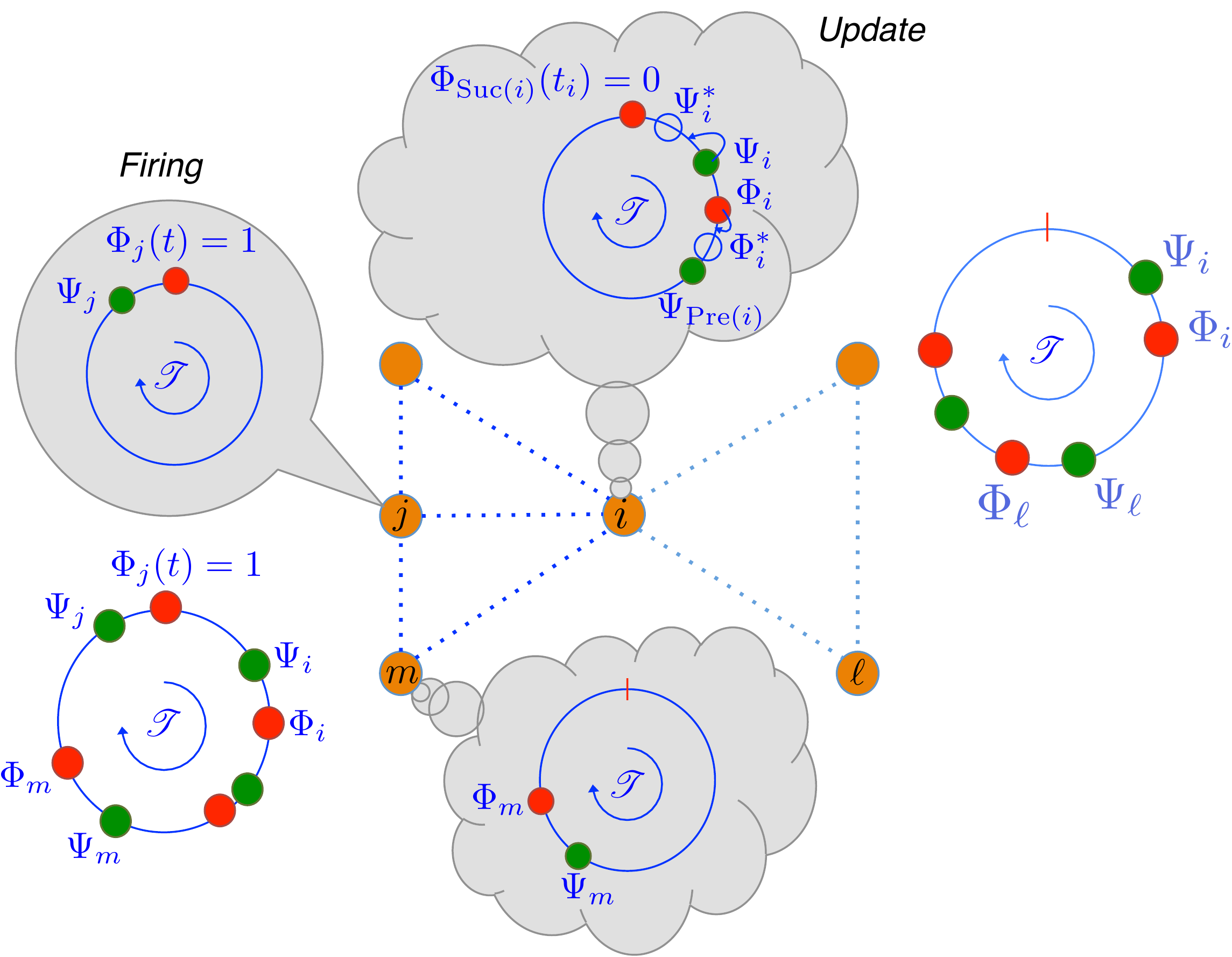}
\caption{Update procedure for the scheduling algorithm when $\Phi_j$ reaches the firing point and node $i$ identifies it as its successor.}\label{fig:PFS_initial_state}
\end{figure}
Introducing:
\begin{align}
\Gamma_{\pi_k^{c}}(t)&=\Phi_{\pi_k^{c}}(t)-\Psi_{\pi_k^{c}}(t)&\pmod{1}\label{Gamma_definition}\\
\Theta_{\pi_k^c}(t)&=\Psi_{\pre(\pi_k^{c},c)}(t)-\Phi_{\pi_k^{c}}(t)&\pmod{1}\label{Theta_definition}
\end{align}
for $k=1,\dots,{|\mathcal{V}_c|}$, it is possible to describe mathematically the evolution of the schedule for every clique $c$ by describing the dynamics of the vector:
\begin{equation}\label{eq:system_state_vector_PFS}
\bUpsilon_{c}(t)\triangleq[\Theta_{\pi_1^c}(t),\Gamma_{\pi_1^c}(t),\dots,\Theta_{\pi_{|\mathcal{V}_c|}^c}(t),\Gamma_{\pi_{|\mathcal{V}_c|}^c}(t)].
\end{equation}
Notice that the entries of this vector are the portions of the frame allocated to each node at time $t$ and the corresponding intermediate guard-spaces. 
Therefore the fixed points of the algorithm represent the final schedule, assuming the demands remain unchanged for a sufficiently long period. It is of interest to understand if the schedule will correspond to an efficient use of the bandwidth and this is the aim of the analysis in the next sections.
\subsection{Convergence of the single clique scheduling algorithm}
In this section we recall the convergence result in \cite{round_robin_pagliari} and analyze the convergence rate of the algorithm.
In the case of a single clique $c$ in our graph, due to the updates, the state vector $\bUpsilon_c(t)$ evolves
linearly with system matrix $\textbf{M}^{c}$ is defined as:
\begin{equation}\label{system_matrix_scheduling_update}
\textbf{M}^{c}=\prod_{k=1}^{|\mathcal{V}_{c}|}\boldsymbol{M}_{\pi_k^c}
\end{equation}
and each $\boldsymbol{M}_{\pi_k^{c}}$ is the matrix for the update of node $\pi_k^{c}$ in the clique $c$.
This matrix has the following form:
\begin{equation}
\boldsymbol{M}_{\pi_k^c}=\mathbf{J}^{(2k-2)}\cdot\left[
\begin{array}{cc}
\textbf{U}_{\pi_k^c} & \mathbf{0}_{3\times (2{|\mathcal{V}_{c}|}-3)}\\
\mathbf{0}_{(2{|\mathcal{V}_{c}|}-3)\times 3} & \Ib_{(2{|\mathcal{V}_{c}|}-3)}\\
\end{array}
\right]\cdot \mathbf{J}^{T^{(2k-2)}}
\end{equation}
where $\mathbf{J}$ represents the circular shift matrix:
\begin{equation}\label{circulant_matrix_J}
\mathbf{J}\triangleq\left[
\begin{array}{ccccc}
0&0&\cdots&0&1\\
1&0&\cdots&0&0\\
0&\ddots&\ddots&\vdots&\vdots\\
\vdots&\ddots&\ddots&0&0\\
0&\cdots&0&1&0\\
\end{array}
\right].
\end{equation}
and
\begin{equation}\label{eq:U_i_def}
\textbf{U}_{i}\triangleq\left[
\begin{array}{ccc}
\vspace{0.2cm}
1-\beta\frac{D_{i}+\delta}{D_{i}+2\delta}&\beta\frac{\delta}{D_{i}+2\delta}&\beta\frac{\delta}{D_{i}+2\delta}\\
\vspace{0.2cm}
\beta\frac{D_{i}}{D_{i}+2\delta}&1-\beta\frac{2\delta}{D_{i}+2\delta}&\beta\frac{D_{i}}{D_{i}+2\delta}\\
\beta\frac{\delta}{D_{i}+2\delta}&\beta\frac{\delta}{D_{i}+2\delta}&1-\beta\frac{D_{i}+\delta}{D_{i}+2\delta}
\end{array}
\right].
\end{equation}
The proof in \cite{round_robin_pagliari} shows that for this configuration there exists a unique fixed point:
\begin{equation}\label{eq:fix_point_sing_clique}
\bUpsilon_c^{\star}=\dfrac{\gamma^{c}}{\boldsymbol{D}^c}(\delta,D_{\pi_1^c},\delta,D_{\pi_2^c},\dots,\delta,D_{\pi_{|\mathcal{V}_{c}|}^c})^T
\end{equation}
where $\boldsymbol{D}^{c}=\sum\limits_{k=1}^{|\mathcal{V}_{c}|}D_{\pi_k^c}$ and $\gamma^{c}=\dfrac{\boldsymbol{D}^{c}}{\boldsymbol{D}^{c}+|\mathcal{V}_{c}|\delta}$.
\vspace{0.1cm}

For the specific case of a single clique $c$ with all nodes having the same demand ($D_{\pi_k^c}=D ~ \forall k=1,2,\dots,|\mathcal{V}_{c}|$) it is actually possible to complement the result with an estimate for the rate of convergence. In fact, when the evolution of the system can be modeled with a linear update, as in average consensus algorithms \cite{gossip_proceedings}, it is well known that the rate of convergence can be estimated via the second largest eigenvalue of the system matrix (i.e., convergence towards a fixed point is guaranteed if the highest eigenvalue is equal to $1$ and the others are strictly smaller).
Assuming equal demand for all nodes, it is possible to rewrite the system matrix in \eqref{system_matrix_scheduling_update} as:
\begin{equation}\label{system_matrix_compact}
\textbf{M}^{c}=\left(\left(
\begin{array}{cc}
\textbf{U} & \mathbf{0}_{3\times (2{|\mathcal{V}_{c}|}-3)}\\
\mathbf{0}_{(2{|\mathcal{V}_{c}|}-3)\times 3} & \Ib_{(2{|\mathcal{V}_{c}|}-3)}\\
\end{array}
\right)\cdot\mathbf{J}^2\right)^{|\mathcal{V}_{c}|}
\end{equation}
where the dependence of the block matrix $\textbf{U}$ on the node $i$ has been lost setting an equal demand $D$ for all nodes.
At this point it is possible to derive the exact $2{|\mathcal{V}_{c}|}$-th degree characteristic equation for the product matrix inside the brackets, find an approximation for the second highest solution and then take the ${|\mathcal{V}_{c}|}$-th power of that value to find the second highest eigenvalue for the system matrix $\textbf{M}^{c}$.
We claim:
\begin{proposition}\label{convergence_rate_lemma}
{\it
The second largest eigenvalue of $\mathbf{M}^{c}$  if all the nodes have the same demand $D$ is:
\begin{equation}\label{eq:second_eigenvalue}
|\lambda_2^{c}|\approx1-\dfrac{2\beta \mu \pi^2}{{|\mathcal{V}_{c}|}^2}
\end{equation}
where $0<\beta<1$ is the coupling factor in the update equation in \eqref{eq:s_update},\eqref{eq:e_update} and $\mu=\dfrac{\delta}{D+2\delta}$.}
\end{proposition}

The proof of this proposition is in Appendix \ref{convergence_rate_proof}.
Clearly, the convergence time increases with the number of nodes in the clique and decreases with the values of $\beta$ and $\mu$.
However, augmenting $\mu$ by  increasing the guard time $\delta$ relative to the demand $D$ lowers its efficiency.
Furthermore, in non-ideal conditions in which in the measurement of $\Psi_{\Pre(i)},\Phi_{\Suc(i)}$ are not precise (see footnote \ref{footnote3}) or the local timers are quantized (for more detailed discussion we refer to \cite{PulseSS_description}), aggressively increasing $\beta$ may result in lack of convergence.
While the rate of convergence is indicative of the trends we found in locally connected networks, the study of the fixed points requires appropriate changes, since the presence of shared nodes changes the structure of $\textbf{M}_{\pi_k^c}$, introducing coupling among the sub-cliques. To describe these changes next we need to introduce new quantities, definitions, assumptions and notations, which precede our main convergence result.
Nevertheless, we wish to remark that the result of Proposition \ref{convergence_rate_lemma} has been found, via simulation, to be a good approximation also for the behaviour of multiclique networks, if we consider for $\mathcal{V}_c$ the largest clique of the graph.
\subsection{Convergence of the multi-cliques scheduling algorithm}\label{sec:convergence_multi}
In this section we analyze what are the possible schedules that are fixed points for the algorithm.
\begin{definition}[Partial proportional fairness criterion]\label{def:partial}
{\it
We say a schedule meets a {\it partial proportional fairness criterion} if, once convergence is reached (i.e. $\forall t>t^*$ for some $t^*$), $\forall (i,j)~~~ i\neq j ~~\mbox{if~~} i,j \in \Lcal_c~\mbox{and}~~j=\pre(i,c)$ the following condition is met :
$$
\frac{(\Phi_i(t)-\Psi_i(t))\!\!\!\!\!\mod 1}{D_i}=\frac{(\Phi_j(t)-\Psi_j(t))\!\!\!\!\!\mod 1}{D_j}
$$
}
\end{definition}
\begin{definition}[Global proportional fairness criterion]\label{def:global}
{\it
We say a schedule meets a {\it global proportional fairness criterion} if the two following properties are satisfied once convergence is reached (i.e. $\forall t>t^*$ for some $t^*$):
\begin{enumerate}
\item $\forall (i,j)~~~ i\neq j ~~\mbox{if~~} i,j \in \Lcal_c,$ then:
$$\frac{(\Phi_i(t)-\Psi_i(t))\!\!\mod 1}{D_i}=\frac{(\Phi_j(t)-\Psi_j(t))\!\!\mod 1 }{D_j}
$$
\item $
\forall j  ~~[(\Phi_j(t)-\Psi_j(t))\!\! \mod1]\geq \min\limits_{c: j \in {\mathcal{V}_c}}\frac{D_j}{\sum_{i\in\mathcal{V}_c}(D_i+\delta)}$
\end{enumerate}
}
\end{definition}
The second property indicates that the solution guarantees that every node gets the minimum possible duration among all cliques in its range.
Let us first introduce the following:
\begin{assumption}\label{ass:consecutive}
{\it For every clique ${\cal V}_c$, all the local nodes (i.e. the set ${\cal L}_c$), occupy consecutive portions of the frame.
}
\end{assumption} 
We can then claim the following:
\begin{theorem}\label{2_clique_PFS}
{\it For a network with two cliques, the update rule in \eqref{eq:s_update} and \eqref{eq:e_update} will converge to a unique fixed point $\bUpsilon_c^{\star}$ $\forall c \in {\mathcal C}$ that respects the {\it partial proprtional fairness criterion} in Definition \ref{def:partial} , irrespective of the initial phases of the timers. 
If Assumption \ref{ass:consecutive} holds, the resulting schedule will also respect the {\it global proportional fairness criterion} in Definition \ref{def:global}.
}
\end{theorem}
The proof can be found in Appendix \ref{proof_PFS_convergence}.
\begin{figure}[t]
\centering
\includegraphics[scale=0.38]{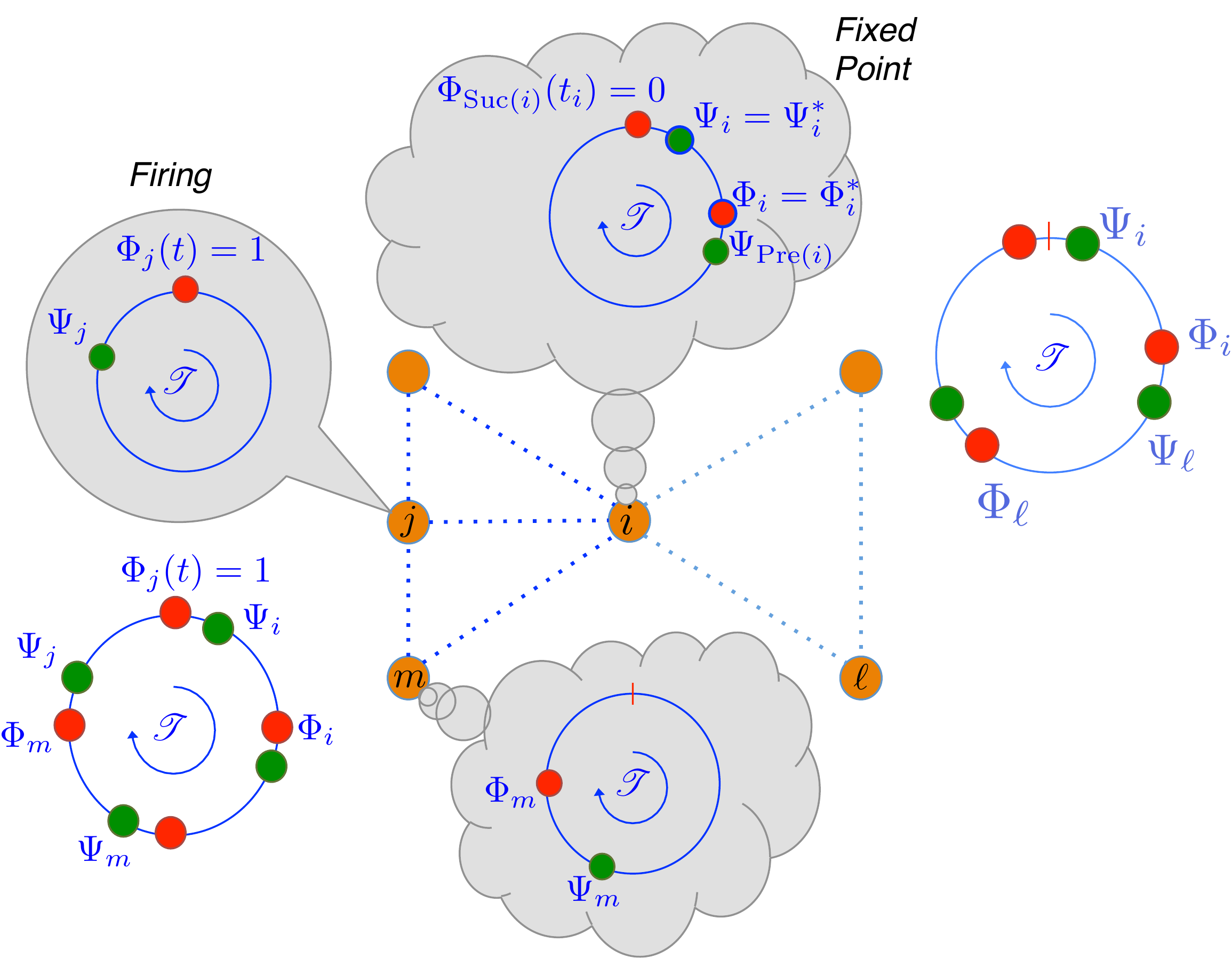}
\caption{Convergence of the scheduling for the topology in Fig.\ref{fig:PFS_initial_state}.}\label{fig:PFS_convergence_state}
\end{figure}

In Fig. \ref{fig:PFS_convergence_state} we can see the convergence of the scheduling algorithm to a fixed point where the target $(\Phi_i^*,\Psi_i^*)=(\Phi_i(t_i),\Psi_i(t_i))$. Therefore, the timers of node $i$ (as well of those of every other node) will no longer change from that update on. 

Following a similar argument as in the proof of Theorem \ref{2_clique_PFS} we can claim  
\begin{proposition}\label{set_multiclique_PFS}
{\it 
For topologies with more than two cliques we have, in general, fixed points for \eqref{eq:s_update}-\eqref{eq:e_update} form sets with measure greater than zero.
All these points respect the partial proportional fairness criterion in Definition \ref{def:partial}.}
\end{proposition} 
See Proof in Appendix \ref{app:proof_set_fixed_points_multiclique}.
In a nutshell, for general topologies we do not have enough contraints on the attainable schedule to guarantee a unique fixed point as in the two-clique case. 
Nevertheless, additional definitions and assumptions allow to characterize very peculiar cases in which a unique fixed point is attainable with more than two cliques. 
Let us introduce the partition $\mathcal{A}_c$ of the nodes set ${\cal V}$ as:
\begin{equation}\label{eq:partition_definition}
{\cal A}_c=\left\{i\left|c=\arg\max_{c'\in{\cal C}_i} \sum_{v\in{\cal V}_{c'}} (D_v+\delta)\right.\right\}
\end{equation}
It is clear that $\forall c,~ {\cal A}_c\subseteq{\cal  V}_c$  and, with distinct overall demands for cliques that have shared nodes (i.e., if $\Scal_{cc'}\neq\emptyset$ then $\sum_{v\in{\cal V}_c} (D_v+\delta)\neq \sum_{v\in{\cal V}_{c'}} (D_v+\delta)$), the sets ${\cal A}_c$ form a proper partition of ${\cal V}$.
Then we order these sets in decreasing order of demand size:
\begin{equation}\label{cluster_sorting}
\sum_{v\in {\cal A}_1} (D_v+\delta)\geq\sum_{v\in {\cal A}_2} (D_v+\delta)\geq\ldots
\sum_{v\in {\cal A}_{|{\cal C}|}} (D_v+\delta)
\end{equation}
Let us introduce: 
\begin{assumption}\label{ass:demand}
{\it
All the nodes in a clique are at most in two partitions $\mathcal{A}_c$ as defined in \eqref{eq:partition_definition}. Mathematically, $\forall c\in\Ccal,$ there is only a single $c'\in \Ccal$ such that 
\be 
\Vcal_c\subset \Acal_{c}\cup\Acal_{c'}
\ee
}
\end{assumption}
The following claim is proven in Appendix \ref{app:proof_unique_fixed_point_multiclique}.
\begin{proposition}\label{unique_multiclique_PFS}
{\it If Assumptions \ref{ass:consecutive}-\ref{ass:demand} are met, the scheduling algorithm and its global proportional fairness property 
can be extended to topologies with an arbitrary number of cliques.}
\end{proposition}
\begin{remark}\label{remark:coloring}
{\it In the limit for  $\delta\rightarrow 0$, if the schedule meets Property 2 in Definition \ref{def:global}, then it is also one of the possible solutions of the minimum coloring graph problem for that conflict graph.}
\end{remark}
To meet Assumption \ref{ass:consecutive} the topology of the conflict graph has to allow an assignment which leaves a portion available in any frame for nodes that belong only to one clique.
However, this assumption may be violated in dense networks, as reported in \cite{PulseSS_description} and a version of Assumption \ref{ass:demand} that explains what conflict graphs can possibly meet Assumption \ref{ass:demand} is elusive.
Nonetheless, the presentation in this work should give the reader the necessary tools to analyze the possible attainable schedules on a case by case basis, given that a general treatment remains elusive.
An example where this assumption is violated and the trend is still predictable is discussed in the proof of Proposition \ref{set_multiclique_PFS} and in the simulation results.
In the next section we provide a description of the specific fixed point $\bUpsilon_c^{\star}$ $\forall c \in {\mathcal C}$ for the case where the demand is equal. For the treatment is advantageous to explicitly indicate the dependence on the frame duration ${\T}$ even though ${\T}=1$.
Note that star and line networks are multi-clique graphs with maximum clique size equal to $2$. For them we can state:
\begin{corollary}\label{corollary_line_star}
{\it The line and the star networks have always a unique fixed point for the schedule consistent with the description in subsection \ref{subsec:fixed_points}.}
\end{corollary}

The proof is in Appendix \ref{corollary_PFS}.
If all the nodes have equal demand $D$, under both of these topologies, $\forall i$:
\begin{equation}
(\Phi_i-\Psi_i)\!\!\!\!\mod{1} \cdot {\T}=\dfrac{D}{D+\delta}\dfrac{{\T}}{2}
\end{equation}
which is about half of the resources as expected.
From Proposition \ref{unique_multiclique_PFS} we can also infer the following:
\begin{corollary}
{\it In the tree network, if parents have always higher demand than the children, there is a unique fixed point consistent with the description in Subsection \ref{subsec:fixed_points}.}
\end{corollary}

This is due to the fact that the condition in this corollary can be seen as an alternative way to state Assumption \ref{ass:demand} for this topology, where we recall that the assignment to every partition ${\cal A}_c$ is based considering the overall demand of nodes (see \eqref{eq:partition_definition}).
This is also a pleasing result, because if the tree is used for data aggregation, it would be natural to have higher demand at the higher level of the tree.
\subsubsection{Fixed points}\label{subsec:fixed_points}
Next we describe what are the unique fixed points attainable under Proposition \ref{unique_multiclique_PFS} (for brevity we assume a single demand value $D_i=D, \forall i\in{\mathcal V}$).
Let us call $\T_c$ the portion of the frame available for the  nodes in $\mathcal{A}_c$. We have: 
\begin{equation}\label{eq:portion_frame_available}
\T_c=\begin{cases}\T&\mbox{if}~~\Acal_c={\cal V}_c\\\T-|{\cal S}_{cc'}|(T_{c'}+\delta_{c'})+\delta_{c'}\hspace{-0.2cm}&\mbox{if}~~\Acal_{c}\!\subset\!{\cal V}_c, \Acal_{c'}\!\supset\!\Vcal_c\setminus\Lcal_c
\end{cases}
\end{equation} 
and Assumption \ref{ass:demand} guarantees the existence of such unique cluster $c'$.
In \eqref{eq:portion_frame_available}, $T_{c'}$ and $\delta_c'$ represent respectively the time slot and the guard space before and after every node $v\in{\cal A}_{c'}$.
They can be computed recursively following the decreasing order in \eqref{cluster_sorting} as follows:
\begin{align}
T_c&=\begin{cases}\dfrac{D}{D+\delta}\dfrac{\T_c}{|{\cal A}_c|}~~&\mbox{if}~~\T_c=\T\label{eq:T_c_fixed_point}\\
\dfrac{D}{D+\delta}\dfrac{\T_c}{|{\cal A}_c|+\frac{\delta}{D+\delta}}~~&\mbox{if}~~\T_c<\T
\end{cases}\\
\delta_c&=\dfrac{\delta}{D}T_c
\end{align}
At the fixed point, $\forall i \in {\cal A}_c$, we will have
\begin{equation}
(\Phi_i-\Psi_i)\!\!\!\!\!\mod 1 \cdot \T=T_c
\end{equation}
To prove this represents a fixed point we will evaluate $\Phi^*_i$ and $\Psi_i^*$ at time $t_i$ (the time when node $i$ makes its update) with \eqref{eq.Ntarget_s},\eqref{eq.Ntarget_e}
(see also Fig.\ref{fig:shared_node_update}):
\begin{align*}
\Phi^*_i(t_i){\T}&=\dfrac{D+\delta}{D+2\delta}\dfrac{D+2\delta}{D}T_c=\dfrac{\delta}{D}T_c+T_c=\Phi_i(t_i){\T}\\
\Psi_i^*(t_i){\T}&=\dfrac{\delta}{D+2\delta}\dfrac{D+2\delta}{D}T_c=\dfrac{\delta}{D}T_c=\Psi_i(t_i){\T}
\end{align*}
So the update procedure will be:
\begin{align*}
\Phi_i(t_i^+)&=(1-\beta)\Phi_i(t_i)+\beta \Phi^*_i(t_i)=\Phi_i(t_i)\\
\Psi_i(t_i^+)&=(1-\beta)\Psi_i(t_i)+\beta \Psi_i^*(t_i)=\Psi_i(t_i)\end{align*}
which shows that the timers will keep their positions unchanged. 
Here we have assumed that $i\in{\cal A}_c$ where ${\T}_c<{\T}$ but the derivation would be exactly the same if we remove the additional term $\frac{\delta}{D+\delta}$ in \eqref{eq:T_c_fixed_point}.
The interest in the derivation of the fixed points is given by their natural connection with the portion of the frame made available to each node.
Numerical results, supporting our theoretical analysis, are next.
\section{Simulation Results}\label{simulation_results}
\subsection{PCO synchronization}
We first show an example in Fig. \ref{fig:pco1} of convergence to the fixed point with and without propagation delays for the PCOs network in Fig. \ref{fig:topology}, where we plot only the components of $\bDelta$ relative to an arbitrary node in $\mathcal{L}_1$. We considered all the distances between the nodes to be equal such that $\tau_{ij}=\tau, \forall i\neq j$.
\begin{figure}
\centering
\includegraphics[scale=0.4]{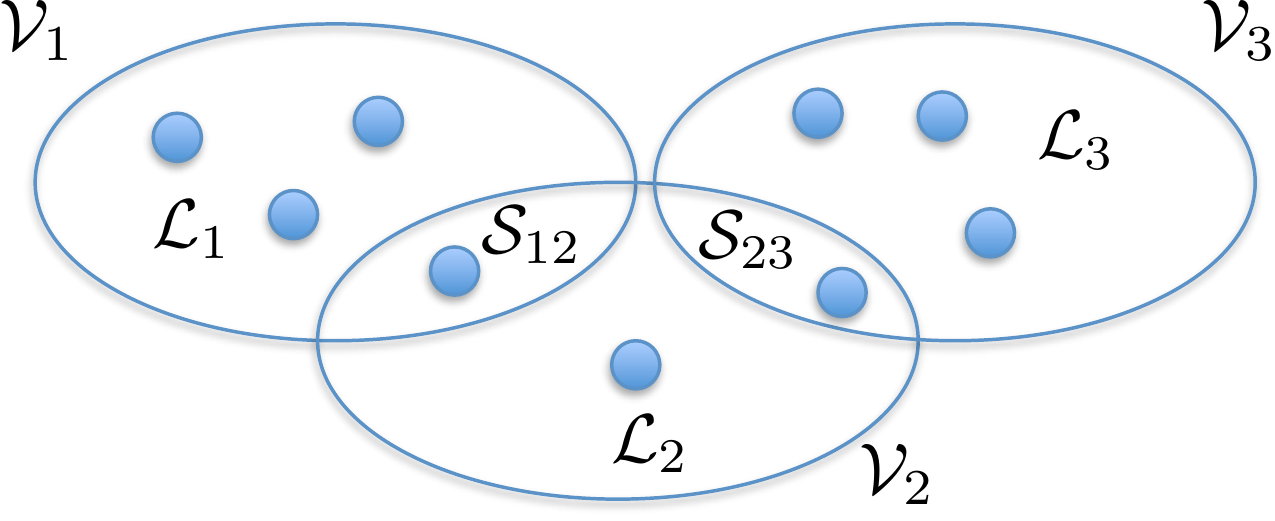}
\caption{Topology considered in the experiments in Fig.\ref{fig:pco1}-\ref{fig:sched1}-\ref{fig:histogram}}\label{fig:topology}
\end{figure}
\begin{figure}
\centering
\subfloat[PCO Evolution without delay]{
\includegraphics[width=0.8\linewidth]{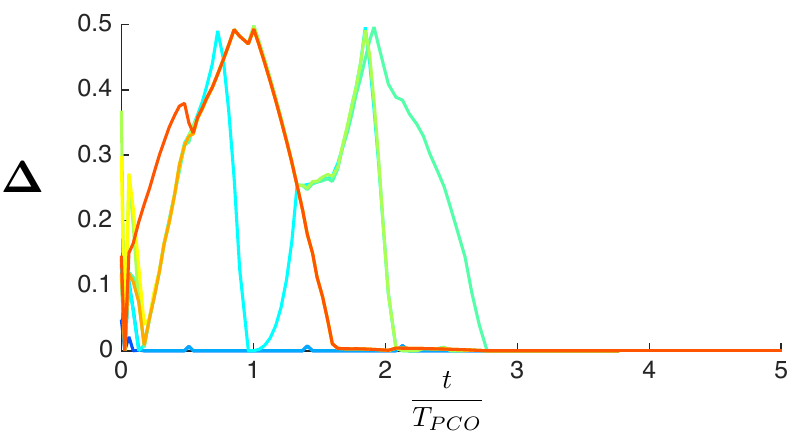}\label{fig:PCO_no_delay}}
\hfill
\subfloat[PCO Evolution with delay]{
\includegraphics[width=0.8\linewidth]{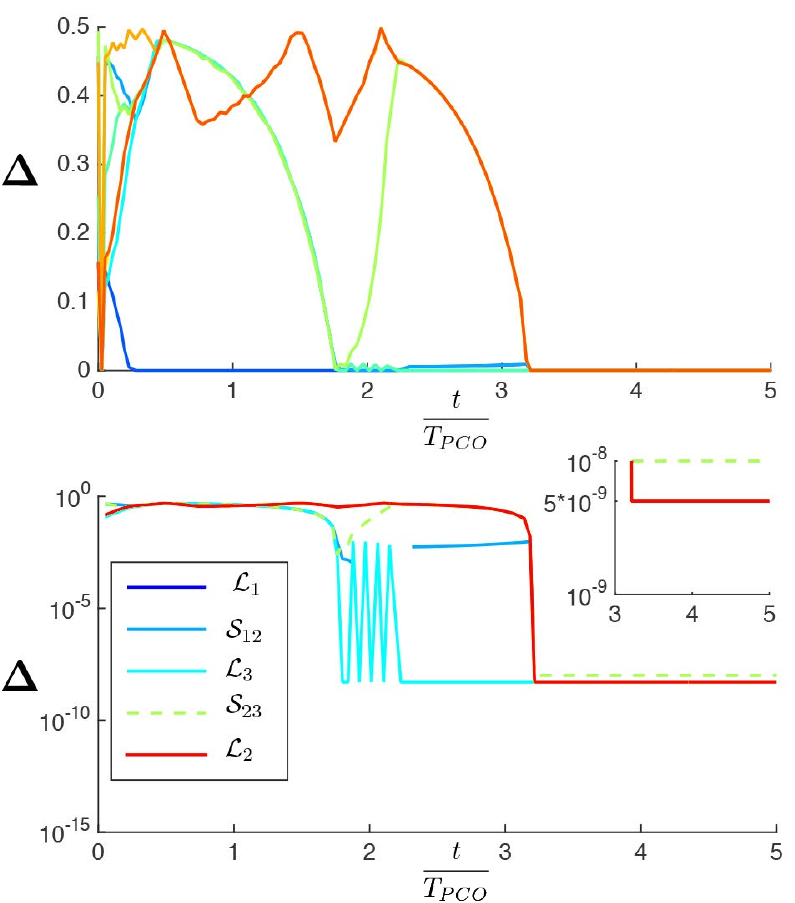}\label{fig:PCO_delay}}
\caption{The evolution of the synchronization with and without delays given the topology seen in (c), and the light blue node as arbitrary chosen reference. Parameters: $T=1s$, $\alpha=1e-2$, distance between each connected node $1m$, signal travel-speed $2e8m/s$, refractory period $1e-2$, uniform random initialization.}\label{fig:pco1}\vspace{-0.5cm}
\end{figure}
We can see that without delays in Fig.\ref{fig:PCO_no_delay} we achieve perfect synchronization, while with delay the nodes remain separated by their propagation delays (see Fig.\ref{fig:PCO_delay}, where we also plot the components in $
\log$-scale.).
In fact, observing the plot in $\log$-scale in Fig. \ref{fig:PCO_delay} we can see that there is a component of $\Delta$ equal to $2\tau$ for the node in $\mathcal{S}_{23}$, the components for the nodes in $\Lcal_2\cup\Lcal_3\cup\mathcal{S}_{12}$ are equal to $\tau$ and the components for the other nodes in $\Lcal_1$ are equal to $0$, from which we can conclude that the {\it head} node is the one node in $\mathcal{S}_{23}$ and all the differences are in perfect agreement with our analysis.
Furthermore, in light of Proposition \ref{fixed_points_theorem_delays} and equation \eqref{residual_sync_error}, we simulated the average residual synchronization error $\mathbb{E}\{\Delta_{\max}\}$ for the line and the star networks with a variable number of nodes.
The probabilities ${p}_h$ for $h=1,2,\dots,N$ are not known, however for any given topology it is possible to bound the residual synchronization error considering the best and the worst case.
In Fig. \ref{line_uniform_tau_deterministic} we show the synchronization accuracy averaged over random initial conditions for {\bf line-networks} with an increasing number of nodes but a constant end to end delay between the nodes at the two network edges. 
The worst case is represented by $\tau_{\max}$ (which is the only possible case for $N=2$), while the best case for a generic $N>2$ is $\tau_{\max}/2$ with the {\it head} node being at the middle of the line.
From Fig. \ref{line_uniform_tau_deterministic} it is possible to notice the saturation of $\mathbb{E}\{\Delta_{\max}\}\approx\frac{3}{4}\tau_{\max}$, half way between the two extremes.

In Fig. \ref{star_uniform_tau_deterministic}, assuming that the delay is proportional to the relative distance, we plot the results for a {\bf star topology} where the nodes are uniformly distributed over a disc with equal radius, leading to equal delay $\tau_{\max}$, except two nodes that are kept fixed at the center and at the edge of the disc. In this case, as expected, increasing the number of nodes degrades the performances. The initial case with $N=2$ is the best case with $\mathbb{E}\{\Delta_{\max}\}=\tau_{\max}$. For this topology the worst case is represented by $2\tau_{\max}$. We notice an oscillation first (due to  $\mathbf{p}_h$) and then, again, a saturation to $\mathbb{E}\{\Delta_{\max}\}\approx\frac{4}{3}\tau_{\max}$.
\begin{figure}
\centering
\subfloat[Line network]{
\hspace{-.2cm}\includegraphics[scale=0.24]{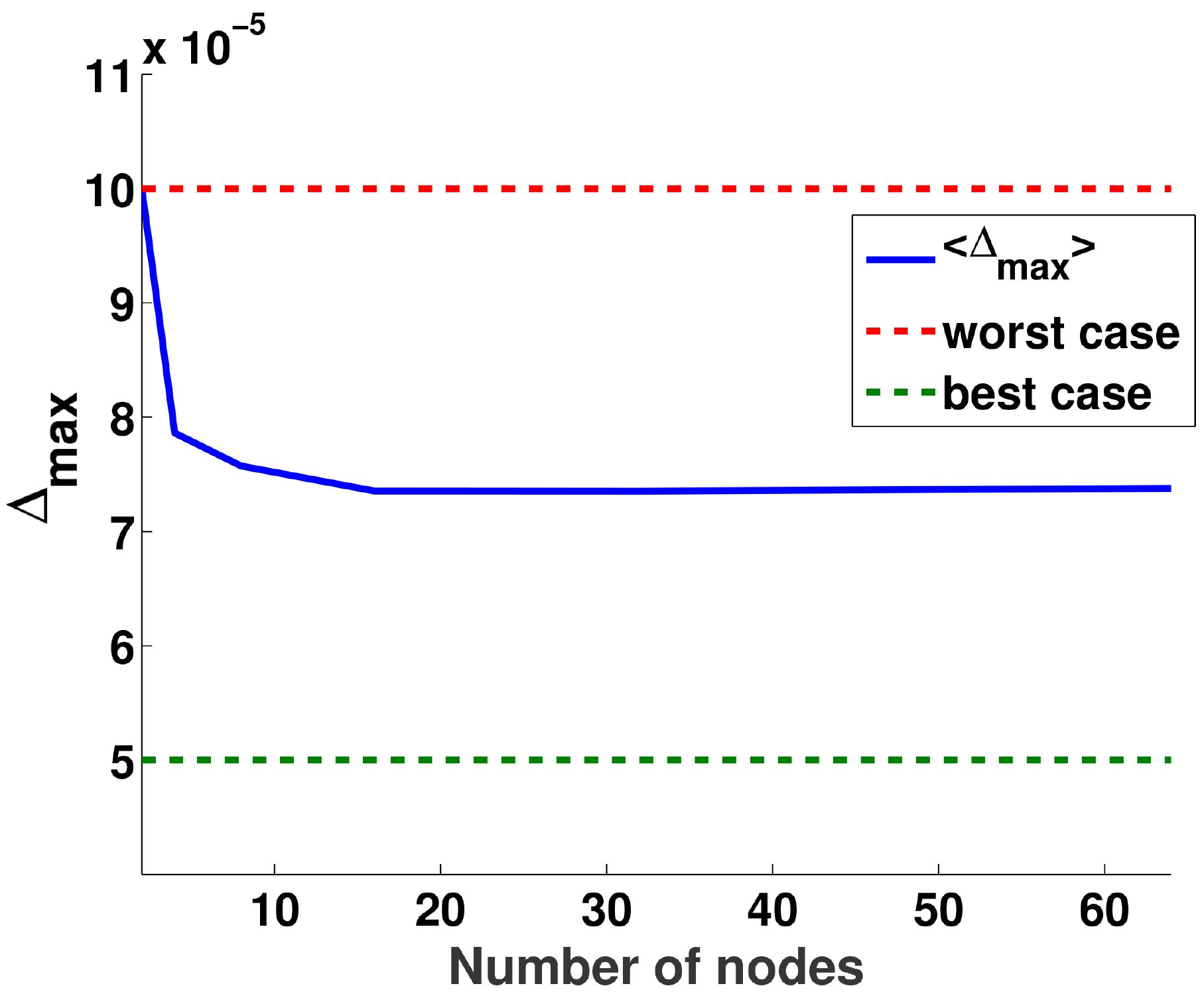}
\label{line_uniform_tau_deterministic}}~\hspace{-.3cm}
\subfloat[Star network]{
\includegraphics[scale=0.24]{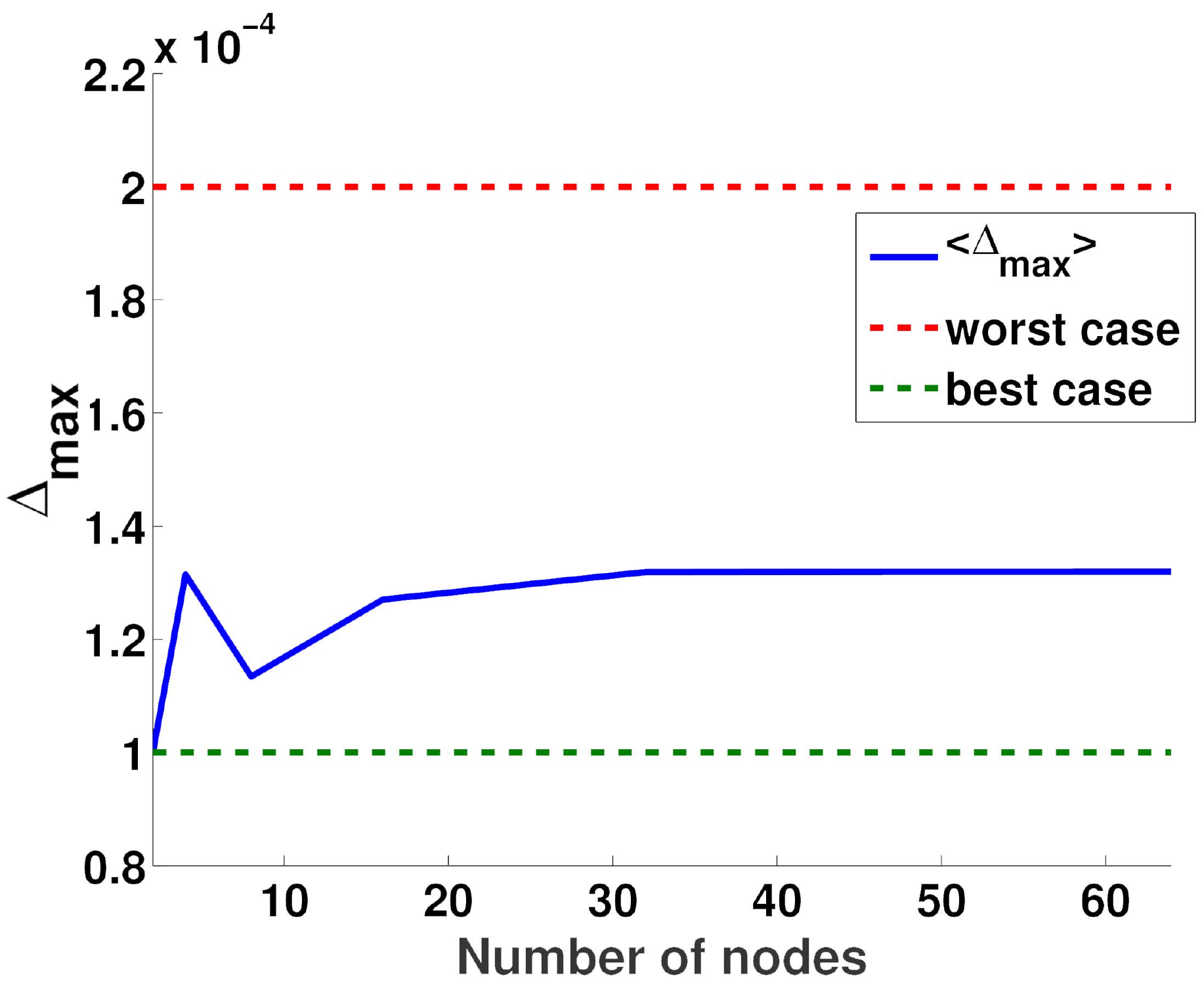}
\label{star_uniform_tau_deterministic}}
\caption{Average maximum displacement vs. network size.}
\label{fig:animals}
\end{figure}

\subsection{Scheduling convergence}
In Fig. \ref{fig:lambda_2} we show the accuracy of the approximate eigenvalues derived in Proposition \ref{convergence_rate_lemma} for the single clique update matrix $\mathbf{M}^c$. The circles correspond to our approximations in \eqref{lambda_equation} and \eqref{lambda_k} (see Appendix \ref{convergence_rate_proof}) and the crosses correspond to the numerical computed values. As expected, the accuracy of our estimate of the second largest eigenvalue grows with $\mathcal{V}_c$.
\begin{figure}
\begin{center}
\includegraphics[scale=0.3]{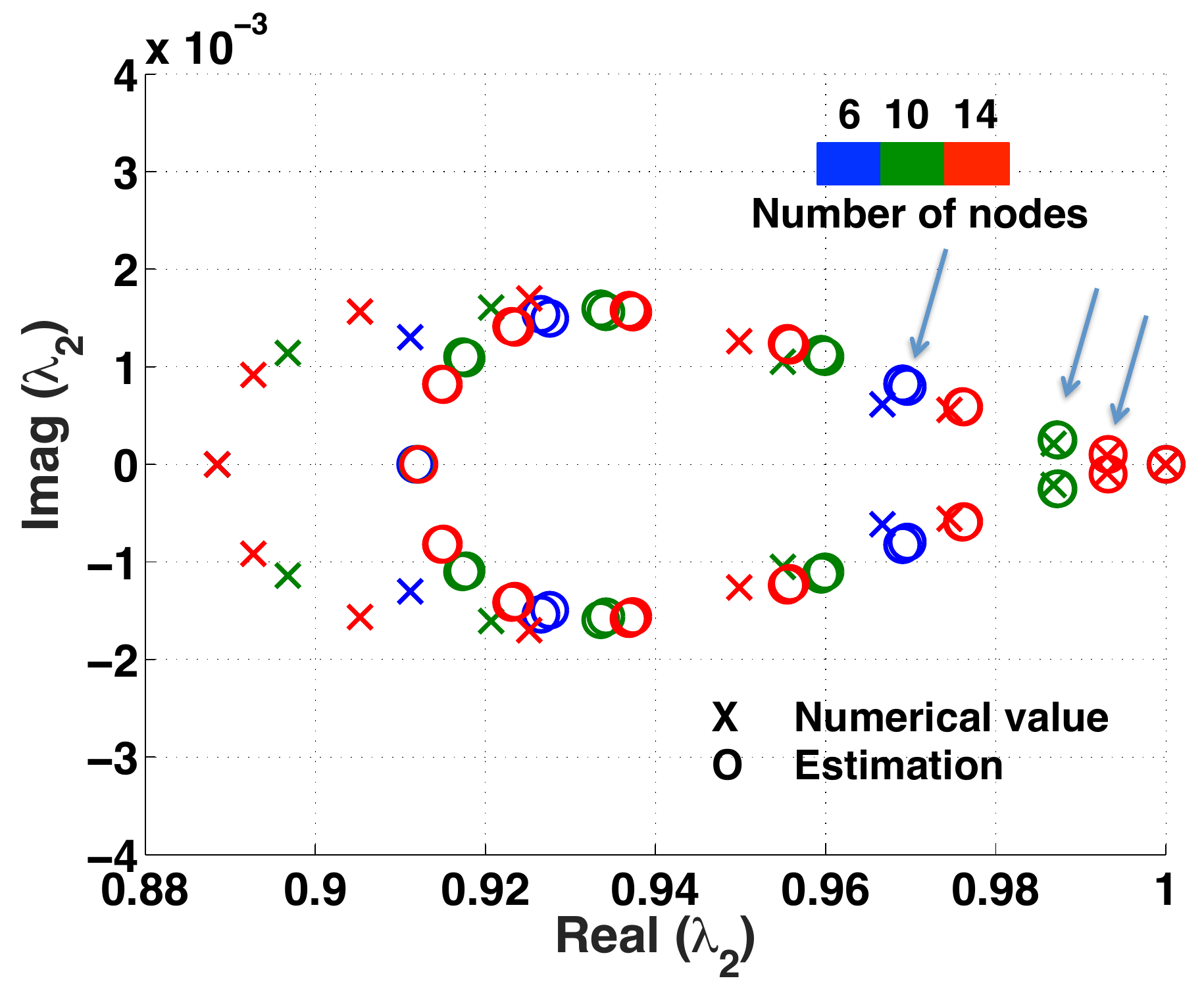}
\end{center}
\caption{Numerical evaluation of the eigenvalues of the matrix $\mathbf{M}^c$ compared with the approximation in \eqref{z_k}.}\label{fig:lambda_2}
\end{figure}
In Fig. \ref{fig:sched2} we present the attainable TDMA scheduling by a two clique topology with $|\Lcal_1|=5$, $|\Scal_{12}|=2$, $|\Lcal_2|=2$ ($D_i=D=4$~$\forall i\in \Vcal$ and $\delta=1$). The local nodes in both $\Lcal_1$ and $\Lcal_2$ occupy consecutive portions of the frame, therefore Assumption \ref{ass:consecutive} is satisfied and in light of Theorem \ref{2_clique_PFS} we have convergence to the unique fixed point that satisfies the {\it global proportional fairness} criterion. In the plot, the start and end timers of each node are shown with the same color, solid lines represent the evolution of the timers and dashed lines represent the predicted fixed point from our analysis in \ref{subsec:fixed_points}. 
\begin{figure}
\subfloat[${\cal V}_1$]{
\includegraphics[scale=0.32]{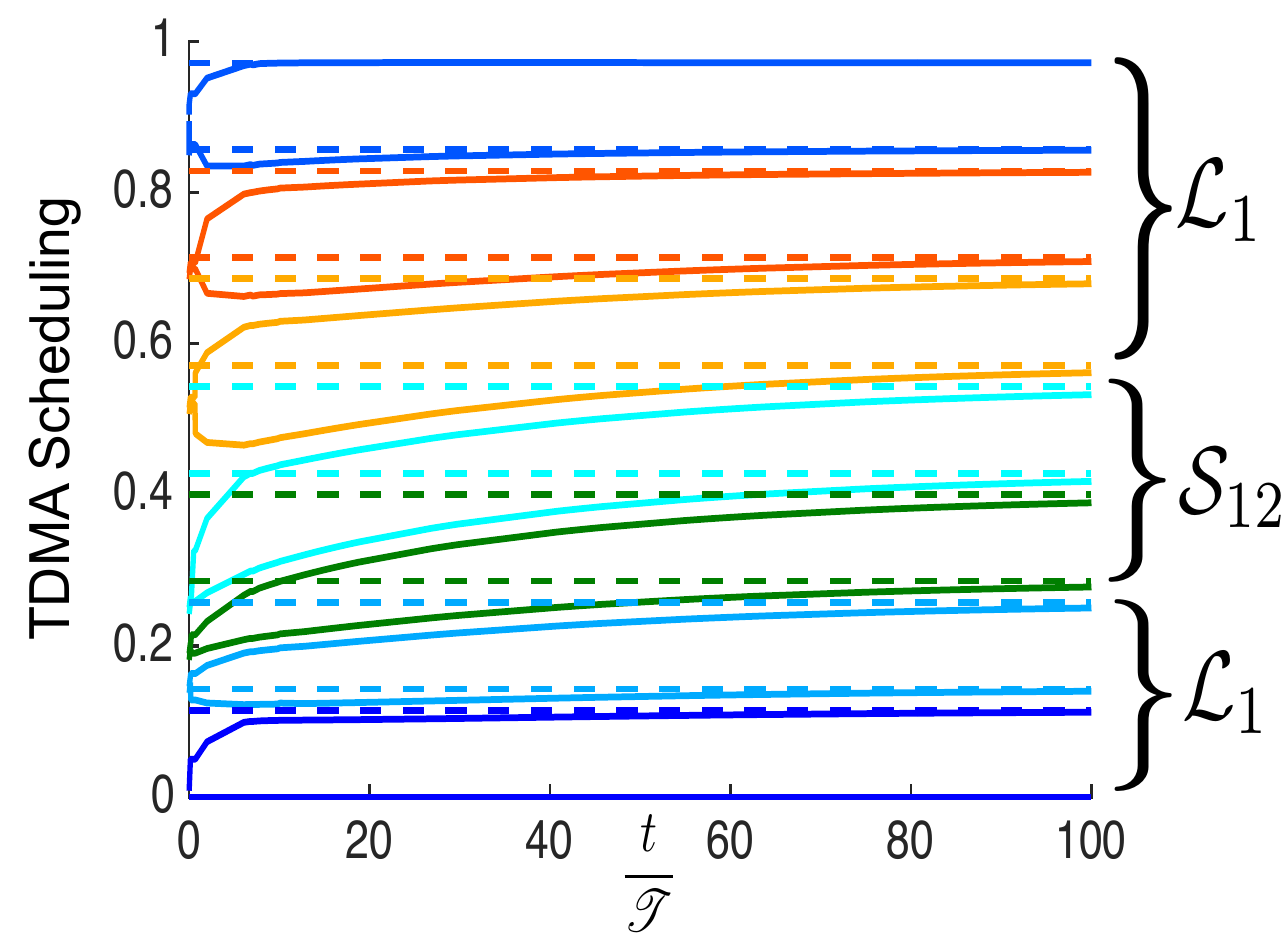}}
\subfloat[${\cal V}_2$]{\includegraphics[scale=0.32]{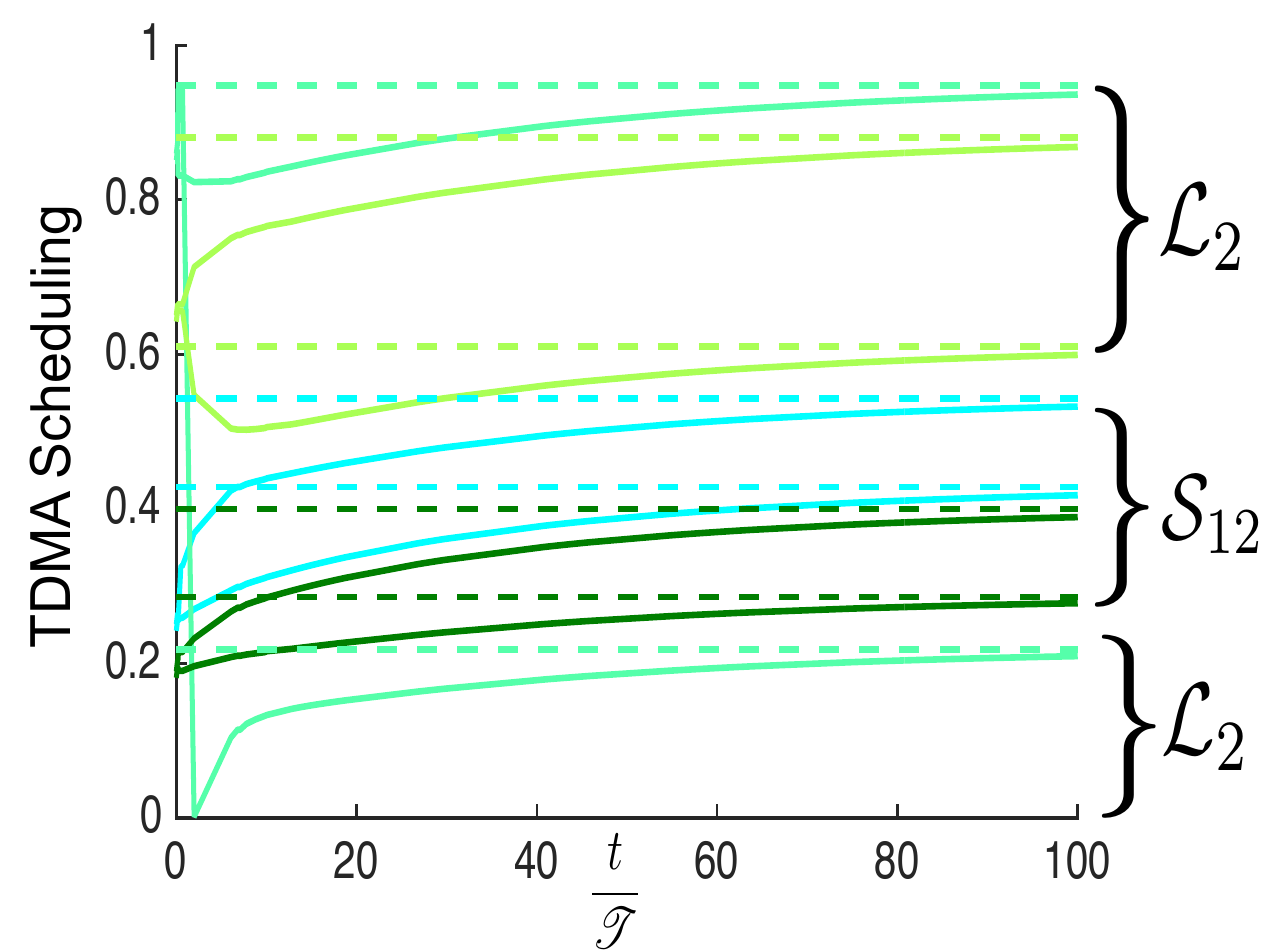}\label{fig:schedb2}}
\caption{TDMA scheduling for a two cliques topology}\label{fig:sched2}
\end{figure}
In Fig. \ref{fig:sched1} we present the case discussed in Appendix \ref{app:proof_set_fixed_points_multiclique} where we have a set of possible fixed points, i.e. a set of attainable schedules. While each node in ${\cal V}_1$ and  in ${\cal V}_3$ reaches its one and only possible schedule, the guard-space between the two shared nodes in ${\cal V}_2$ allows a range of fixed points
 for the local node in  ${\cal V}_2$. The range is limited on both sides such that the local node in ${\cal V}_2$ is never the predecessor or successor of any of the two shared nodes. In Fig. \ref{fig:schedb} we show the range of possible start beacons $\Phi$ with two dashed lines and the range of end beacons $\Psi$ with a dash-dotted line of the same color in the simulation. In Fig. \ref{fig:histogram} we plot the histogram of different shares obtained by the node in $\Lcal_2$ obtained by MonteCarlo simulations. We can see the range is the one predicted by our equations in Appendix \ref{app:proof_set_fixed_points_multiclique}. A pleasant result is that in larger number of cases ($\sim 43\%$) the local node gets the maximum share possible, i.e. {\it global proportional fairness} is often obtained, even though the conditions in Assumption \ref{ass:demand} are not met.
\vspace{-.3cm}
\begin{figure}
\subfloat[${\cal V}_1$]{
\includegraphics[scale=0.3]{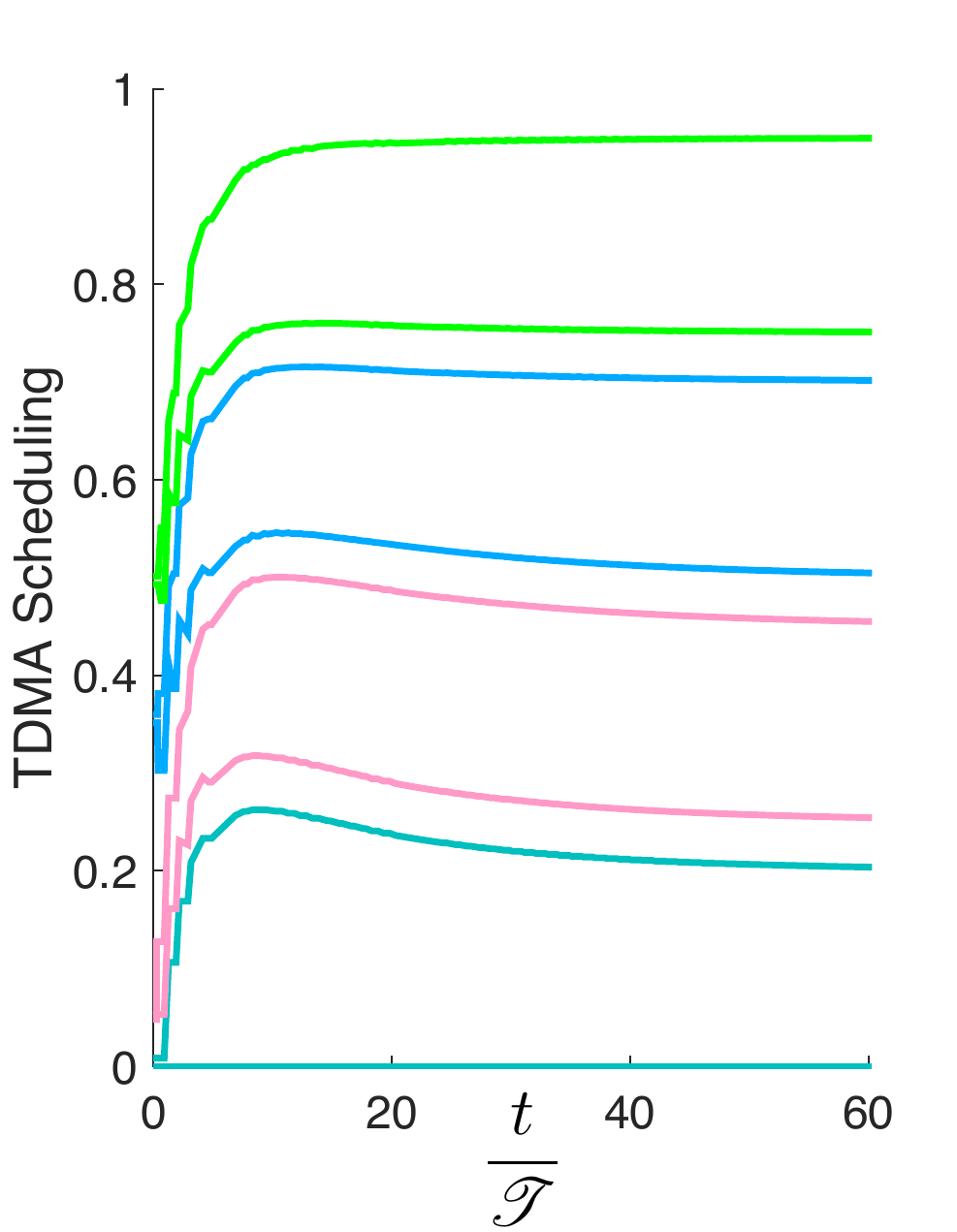}}
\subfloat[${\cal V}_2$]{\includegraphics[scale=0.3]{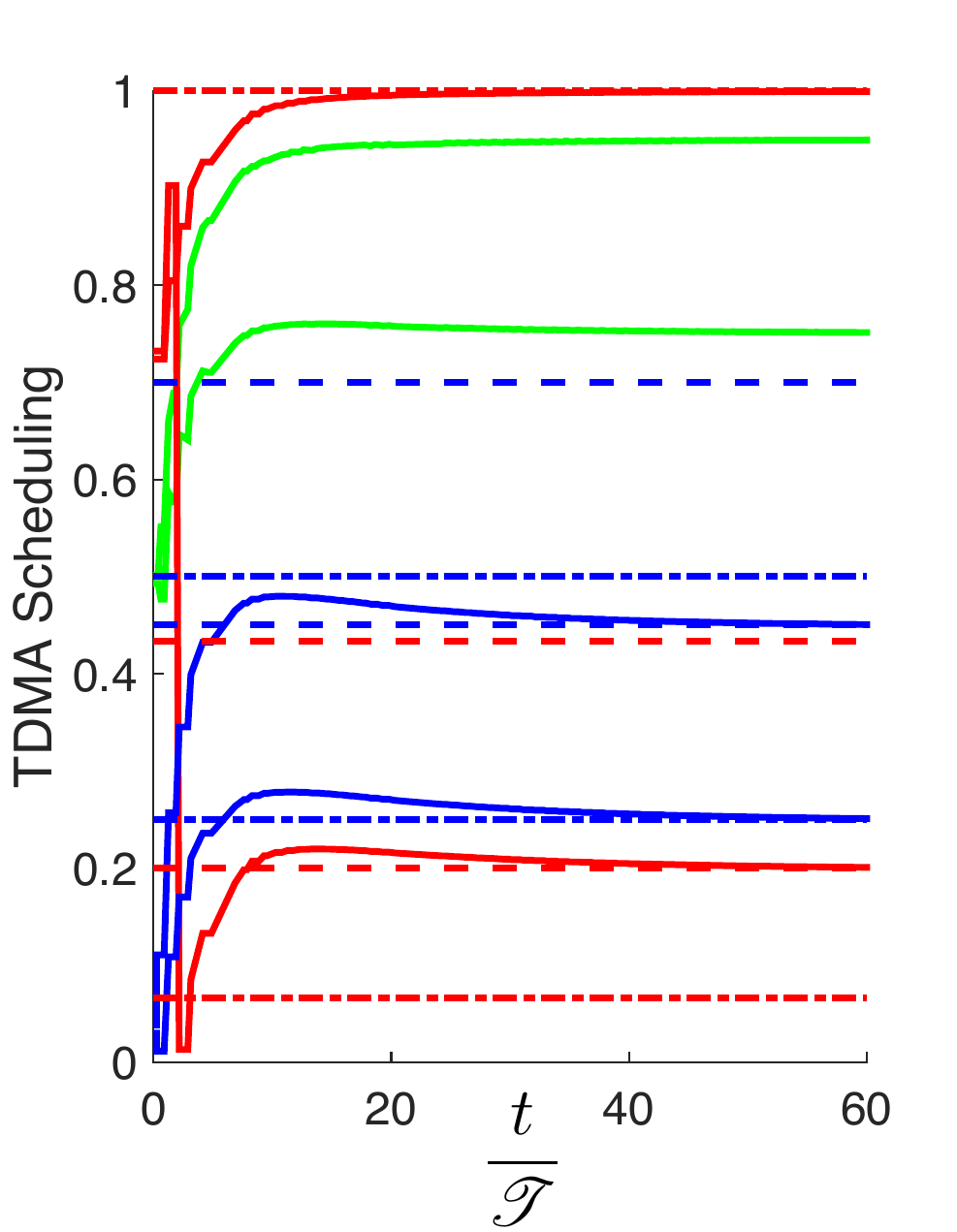}\label{fig:schedb}}
\subfloat[${\cal V}_3$]{
\includegraphics[scale=0.3]{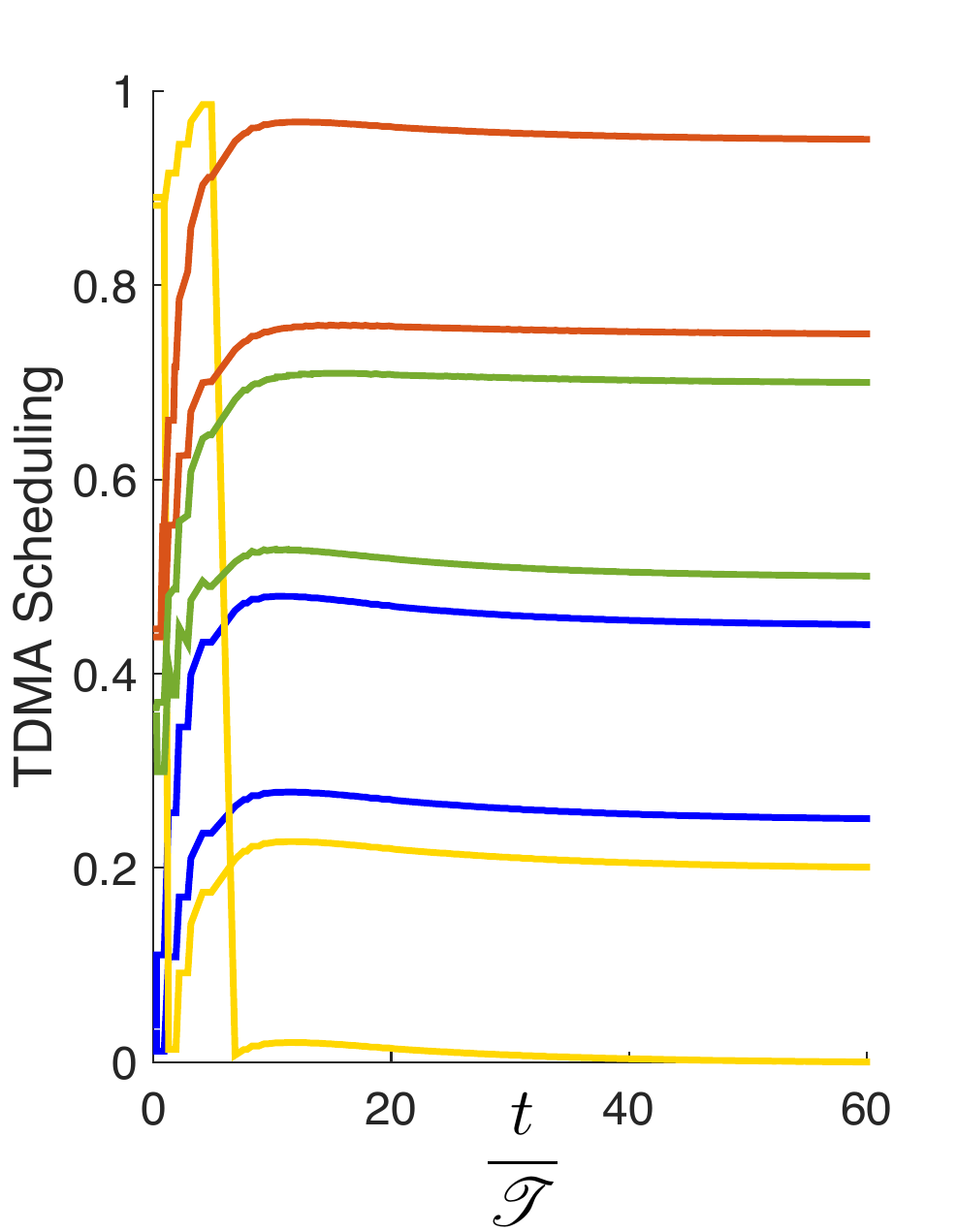}}
\caption{TDMA scheduling of the topology in Fig. \ref{fig:topology}. Note that the plots a-b-c are circular, thus `1' is adjacent to `0', and thus the `jump' of the yellow node in ${\cal V}_3$.}\label{fig:sched1}
\end{figure}
\begin{figure}
\centering
\includegraphics[scale=0.3]{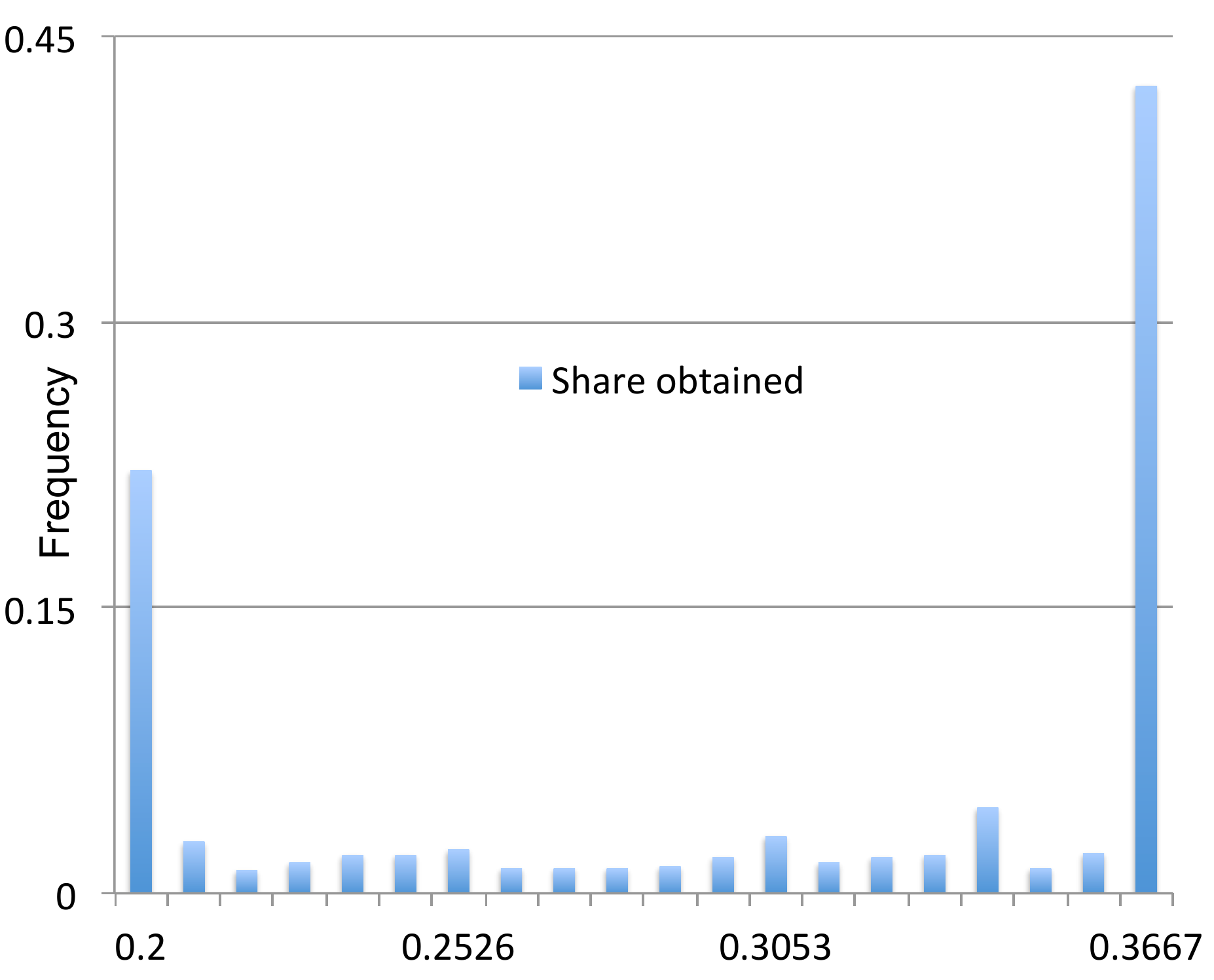}
\caption{Transmission share obtained by the local node in $\Lcal_2$ at convergence, in the range predicted in Appendix \ref{app:proof_set_fixed_points_multiclique}.}\label{fig:histogram}
\end{figure}
\section{Conclusions}
In this work, further results regarding the convergence of synchronization and scheduling in locally connected networks were derived. In particular, for PCO synchronization, the convergence was shown to occur almost surely for a network of $3$ nodes and was characterized in terms of final accuracy for larger networks. For PCO scheduling, the convergence is shown for networks with mild conditions on the overlapping set of maximal cliques. The findings were sustained with the numerical results, taking into account both accuracy and convergence. 
This work advances the theory on the convergence of PCOs in locally connected networks.
\vspace{-.2cm}
\begin{appendix}
\vspace{-.1cm}
\subsection{Proof of Proposition \ref{PCO_no_delays_fixed_point}}\label{proof_PCO_no_delays_fixed_point}
We first show $\bDelta(t^*)=0\Rightarrow \forall t>t^*~~\bDelta(t)=\bDelta(t^*)$ which indicates the synchronous state is a fixed point. 
The argument is straightforward, since all the firings events are received instantaneously by nodes that are also firing, i.e. their phase is equal to $0\mod{1}$ and it is immediate to see no node whill change its phase, i.e. $\min\{(1+\alpha)0\mod{1},1\}=0\mod{1}$ therefore no change in $\bDelta$ will occur from this point on.
Then we prove $\forall t>t^*~~\bDelta(t)=\bDelta(t^*)\Rightarrow \bDelta(t^*)=0$ by contradiction.
Let us assume $\Delta_{ij}(t^*)\neq 0$ for some $i,j$. Without loss of generality we can assume $e_{ij}=1$, in fact if $\Delta_{ij}\neq 0$ and $e_{i,j}\neq 0$ then, by looking at the edges $k\ell$ over the path ${\cal P}_{ij}$ we must find a nonzero phase difference $\Delta_{k\ell}\neq 0$, $e_{k\ell}=1$.
Then, without loss of generality let us consider the firing of node $i$ heard by node $j$. Node $j$ will update its phase and $\Delta_{ij}$ will either increase or decrease, unless the phase of node $j$ is $0$ when node $i$ is firing, which would contradict the hypothesis that $\Delta_{ij}\neq 0$.  
We can consider the isolated event ``node $j$ hears node $i$ firing and updates its phase'', since each other event that occurs simultaneously will not change the phase of node $i$ (which is equal to $0\mod{1}$) and will only potentially move even forward node $j$, further increasing or decreasing $\Delta_{ij}$. 
\subsection{Proof of Proposition \ref{PCO_conv_3_nodes}}\label{proof_PCO_conv_3_nodes}
In order to prove the almost sure convergence of a $|\mathcal{V}|=3$ nodes network to $\bDelta={\bf 0}$, let us label the center node as $1$.
It is easier to focus on the evolution of the variables
\be
\boldsymbol{\Xi}\triangleq\{\Xi_{12},\Xi_{13}\},~~~\Xi_{ij}=(\Phi_i-\Phi_j)\!\!\!\mod{1} 
\ee
and, since by definition \eqref{eq:delta_1_definition} we have $\Delta_{ij}=\min\{\Xi_{ij},\Xi_{ji}\}$, if we have convergence to $\mathbf{0}$ (synchronization) for $\boldsymbol{\Xi}$ we also have it for $\bDelta$. We focus on the line network with nodes $\{1,2,3\}$ and edges $\{(1,2),(1,3)\}$ since the case of a fully connected network is covered by \cite{strogatz}. Since the evolution of $\bXi(t)$ occurs in jumps that are triggered by the firing events, we can define $\boldsymbol{\Xi}[k]\triangleq \bXi(t_{f[k]})$, where $f[k]$ is the index of the node that is generating the $k$-th firing and $t_{f[k]}$ is the time for which the $k$-th firing occurs and focus on the evolution of $\bXi[k],\bXi[k+1],\dots$.

We highlight that for the case of $3$ nodes (or any tree network for that matter) it is possible to consider each firing event separately, since the only case for which a node is affected by two simultaneous firing events is when node $1$ can hear node $2$ and node $3$.
To handle this case we first introduce
\begin{lemma}\label{edge_sync_lemma}
{\it
Suppose $\Phi_2(t)=\Phi_3(t)$ and, equivalently, $\Xi_{12}(t)=\Xi_{13}(t)$. Then, $\Xi_{12}(t')=\Xi_{13}(t')$, for all $t'\geq t$, regardless of the sequence of firing events occuring after $t$.}
\end{lemma} 
\begin{proof}
The statement follows from the fact that at the firing events of node $1$, the two nodes update simultaneously and change their phase by the same amount, while when one of the two fires, the other one is also at the firing point and they do not affect each other.
\end{proof}
Lemma \ref{edge_sync_lemma} also implies that when node $2$ and $3$ are synchronized, we have equivalently a fully connected two-node network whose convergence occurs for all initial values of the nodes' phases, except for a set of measure zero, as shown in \cite{strogatz}. Hence, the network converges almost surely to the fixed point $\bDelta={\bf 0}$ in this case.
We can then proceed in our proof treating each firing event separately, and in particular consider two cases: the case where the firing order is maintained and the case where overtaking of the firing order may occur among nodes. We initially assume all nodes have different phases.
\underline{Case 1:} Suppose that the firing order does not change after some $k_0$ (the $k_0$-th firing event) and that the nodes are labeled in the order of their firing after this point, with node $1$ being the node firing in the middle. That is, suppose that, for some $k_0$, $\Xi_{12}[k]<\Xi_{13}[k]$, for all $k\geq k_0$. If $f[k]=1$ and no absorption occurs then, we have $\Xi_{13}[k+1]-\Xi_{12}[k+1]=(1+\alpha)(\Xi_{13}[k]-\Xi_{12}[k])>\Xi_{13}[k]-\Xi_{12}[k]$, that is, the phase difference between nodes $3$ and $2$ increases.
If $f[k]=2$ or $3$ and still no absorption occurs we can see that the phase difference between nodes $3$ and $2$ remains the same (i.e., $\Xi_{13}[k+1]-\Xi_{12}[k+1]=\Xi_{13}[k]-\Xi_{12}[k]$). This implies that, if no absorption occurred after each node has fired, the phase difference between nodes $3$ and $2$ must strictly increase, that is, $0<\Xi_{13}[k]-\Xi_{12}[k]<\Xi_{13}[k+3]-\Xi_{12}[k+3]$. Since $\Xi_{13}[k']-\Xi_{12}[k']<1$, for all $k'$, an absorption must eventually occur in one of the cycles. If absorption occurs between nodes $2$ and $3$, then we are done since, by Lemma \ref{edge_sync_lemma}, we have $\Xi_{12}[k]=\Xi_{13}[k]$ from that point on. If node $3$ absorbs node $1$, then the next one to fire is node $2$ and that will trigger node $1$ to overtake the position of node $3$, which violates our assumption in Case 1. Finally, if node $1$ absorbs only node $2$ at some firing event $k''$ (that is, if $\Xi_{12}[k'']=0$ and $\Xi_{13}[k'']>0$), then the firing of node $3$, which comes immediately after will cause either an absorption of both nodes $1$ and $2$ (which occurs if $\Xi_{13}[k'']\in[\frac{1}{1+\alpha},1)$) or an update of $\Xi_{12}[k''+1]=\alpha \Xi_{13}[k'']$, if $\Xi_{13}[k'']\in(0,\frac{1}{1+\alpha})$. In the former case, the nodes become synchronized and we are done; in the latter case, it must be true that $\Xi_{12}[k''+1]<\frac{\alpha}{1+\alpha}$ and, thus, node $1$ will again absorb node $2$ when it fires. Therefore, even though the phases of nodes $2$ and $1$ may temporarily deviate from each other, they will always become absorbed again once node $1$ fires. In this case, we  again have an equivalent two-node fully connected network (formed by node $3$ and the combination of nodes $1$ and $2$) and, thus, the convergence to the fixed point $\bDelta={\bf 0}$ again follows from \cite{strogatz}. \underline{Case 2:} Suppose that a node may overtake the position of another node in the firing order. This can occur only when the firing of node $2$ triggers node $1$ to increase its phase beyond the phase of node $3$ (causing node $1$ to fire again before node $3$ fires). This is because, when node $3$ fires (i.e., $f[k]=3$), either node $1$ (and, maybe also node $2$) is absorbed (in which case the firing order is considered to be maintained) or no node is absorbed and the following update occurs: $\Xi_{13}[k+1]=(1+\alpha)\Xi_{13}[k]>\Xi_{12}[k]+\alpha \Xi_{13}[k]=\Xi_{12}[k+1]$, which also implies that the firing order remains the same; and, when node $1$ fires, we again have $\Xi_{13}[k+1]\geq \Xi_{12}[k+1]$ and, thus, the firing order is again maintained. Now, suppose that the firing of node $2$ causes node $1$ to overtake the position of node $3$ (that is, $\Xi_{13}[k]\geq \Xi_{12}[k]$ but $\Xi_{13}[k+1]<\Xi_{12}[k+1]$).
This can occur only when $\Xi_{12}[k]\in(0,\frac{1}{1+\alpha})$; otherwise, the firing of node $2$ would have absorbed both nodes $1$ and $3$ (in which case the convergence to the fixed point is immediately achieved). However, if overtaking occurs due to the firing of node $2$, then this means that $\Xi_{13}[k+1]=\Xi_{13}[k]+\alpha\Xi_{12}[k]\!\mod{1}<\Xi_{12}[k+1]=(1+\alpha)\Xi_{12}[k]$. Since $\Xi_{13}[k]\geq \Xi_{12}[k]$, this is possible only when $\Xi_{13}[k]+\alpha\Xi_{12}[k]\geq 1$.
Therefore, $\Xi_{13}[k+1]=\Xi_{13}[k]+\alpha\Xi_{12}[k]\mod{1}=\Xi_{13}[k]+\alpha\Xi_{12}[k]-1<\frac{\alpha}{1+\alpha}$ since $\Xi_{13}[k]< 1$ and $\Xi_{12}[k]\in(0,\frac{1}{1+\alpha})$. In this case, the firing of node $1$, which comes immediately after, will again absorb node $3$, causing the two nodes to be absorbed from that point on. Even though the firing of node $2$ may trigger node $1$ to temporarily overtake node $3$, they will always become absorbed again once node $1$ fires. Consequently, we again have an equivalent two-node fully connected network (formed by node $2$ and the combination of nodes $1$ and $3$); convergence to the fixed point $\bDelta={\bf 0}$ follows from \cite{strogatz}.
In light of this extensive treatment, we can then consider the case where a node is added to a synchronized network. At the time this additional node fires, the nodes connected to it will jump in front of the others but, since no additional firing events will occur before the neighbors of the new node node fire, all the absorption with the other nodes will be restored at the end of the firing round and so we have again an equivalent two-node fully connected network (formed by the added node and the previous synchronized network) and convergence to synchronization will occur almost surely as previously discussed.
\vspace{-.3cm}
\subsection{Proof of Proposition \ref{fixed_points_theorem_delays}}\label{proof_theorem_delays}
We first prove $\forall t>t^*, \bDelta(t)=\bDelta(t^*)\Rightarrow \bDelta(t^*)\in\Fcal $ by contradiction, i.e. by showing that any point $\bDelta$ for which $\Delta_{ij}>\tau_{ij}$ for some $ij$ s.t. $e_{ij}=1$, cannot be a fixed point.
We can initially assume, without loss of generality that $(\Phi_i-\Phi_j\!\!\mod{1})<(\Phi_j-\Phi_1\!\!\mod{1})$  and look at what happens when node $i$ fires.
At the time node $i$ fires, say $t_i>t^*$, since $\Delta_{ij}(t_i)>\tau_{ij}$ and $(\Phi_i-\Phi_j\!\!\mod{1})<(\Phi_j-\Phi_1\!\!\mod{1})$ we have that $\frac{1}{2}<\Phi_j<1-\tau_{ij}$.
Let us initially assume nothing happens (no firing events) before node $j$ hears the firing of node $i$ at time $t_i+\tau_{ij}$.
The phase of node $j$ at time $t_i+\tau_{ij}$ will be $\frac{1}{2}+\tau_{ij}<\Phi_j(t_i+\tau_{ij})<1$ and therefore outside of the refractory period, as long as $\rho<\frac{1}{2}+\min_{ij}\tau_{ij}$. This implies that node $j$ will jump forward and the difference $\Delta_{ij}$ will decrease, which means it is not a fixed point. 
If some firing events happen before the update, i.e at time $t$ where $t_i<t<t_i+\tau_{ij}$ then node $i$ will not move since its phase will be $\Phi_i(t)<\tau_{ij},\forall t_i<t<t_i+\tau_{ij}$ and therefore it is inside the refractory period $\rho>2\max_{ij}\tau_{ij}$. Node $j$ instead, can either move even forward by hearing these additional events (and $\Delta_{ij}$ would further decrease, therefore keeping our argument by contradiction valid) or not move (and therefore there is no problem in neglecting such events), and this concludes our proof. 
The  key aspect of the proof is that nodes can only jump forward after receiving a firing event. 
We then prove $\bDelta(t^*)\in\Fcal \Rightarrow \forall t>t^*, \bDelta(t)=\bDelta(t^*)$ which indicates that all the points in $\Fcal$ are fixed points. Let us consider a general node $i$ and show it will note update its phase. In order to update its phase, node $i$ needs to receive one or more firing events by its neighbours when its phase is outside the refractory period.
Since $\Delta_{ij}\leq\tau_{ij}$ $\forall j, e_{ij}=1$ there are then $2$ cases.\\
\underline{Case 1:} A set $J$ of node $i$'s neighbours fire at time $t'$ when $0\leq\Phi_{i}(t')\leq\min_{j\in J}\tau_{ij}$. 
All these firings will be heard by node $i$ when $\min_{j\in J}\tau_{ij}\leq\Phi_i\leq\min_{j\in J}\tau_{ij}+\max_{j\in J}\tau_{ij}\leq 2\max_{ij}\tau_{ij}\leq\rho$, therefore no update will occur.\\
\underline{Case 2:} A set $J$ of node $i$'s neighbours fire at time $t'$ when $1-\min_{j\in J}\tau_{ij}\leq\Phi_{i}(t')\leq 1$. 
All these firings will be heard when $(1-\min_{j\in J}\tau_{ij}+\min_{j\in J}\tau_{ij})\mod{1}\leq\Phi_i\leq\max_{j\in J}\tau_{ij}\rightarrow ~0\leq\Phi_i\leq\max_{j\in J}\tau_{ij}<\rho$, therefore no update will occur.
If no node updates its phase, then $\bDelta(t)$ remains constant over time and this proves that all $\bDelta\in\Fcal$ are fixed points. 
\vspace{-.2cm}
\subsection{Proof of Proposition \ref{convergence_rate_lemma}}
\label{convergence_rate_proof}
We know $\textbf{M}^{c}$ is a stochastic matrix and then all of its eigenvalues are inside the unit circle except for one.
After a few tedious but straightforward manipulation the $2n$-degree characteristic equation of $\textbf{M}^{c}$ is ($n=|{\mathcal{V}}_c|$):
\begin{align}
&\lambda^{2n}-\lambda^n-(\beta-1)^2(\lambda^n-1)+\beta^2\mu(2\lambda^n-\lambda^{n-1}-\lambda)\nonumber \\
&-\beta \mu(\lambda^{2n-1}+\lambda^{n+1}-\lambda^{n-1}-\lambda)=0\label{lambda_equation}.
\end{align}
This specific form highlights that $\lambda=1$ is solution and, also, that for $\mu=0$ $\lambda^n=1$.
For $\mu>0$ small, the $n-1$ roots of \eqref{lambda_equation} inside the unit circle can be approximated as:
\begin{equation}\label{lambda_k}
\lambda(k)=(1-\varepsilon)e^{j(\frac{2\pi k}{n}-\vartheta)}
\end{equation}
with $k=1,2,....,n-1$ and here $j=\sqrt{-1}$ indicates the imaginary unit. We are interested in the second largest eigenvalue of the system, thus in the highest of these $n-1$ perturbed roots.
For $\mu$ small enough we assume that $\varepsilon,\vartheta\ll 1$.
Then it is possible to substitute \eqref{lambda_k} in \eqref{lambda_equation} and consider:
\begin{equation*}
\left((1-\varepsilon)e^{-j\vartheta}\right)^n\approx (1-n\varepsilon-nj\vartheta)
\end{equation*}
We are then able to solve \eqref{lambda_equation} as a first order equation in $z=\varepsilon+j\vartheta$ and then consider the real part of the solution ($\varepsilon$) and the imaginary part ($\vartheta$) to find the eigenvalues from \eqref{lambda_k}.
The solutions will have the following form:
\begin{footnotesize}
\begin{equation}\label{z_k}
z^*(k)=\dfrac{1}{n}\dfrac{\left(1-\cos{\frac{2\pi k}{n}}\right)}{1-\frac{1}{2}\exp\{-j\frac{2\pi k}{n}\}-\frac{1}{n}\sin\frac{2\pi k}{n}+\frac{1}{\beta\mu}(1-\frac{\beta}{2}-\mu\cos\frac{2\pi k}{n})}
\end{equation}
\end{footnotesize}
It is possible to show that the real part of $z^*(k)$ is concave down with respect to $k$, and thus from \eqref{lambda_k} the second largest eigenvalue is reached for $k=1$ or $k=n-1$.
A simple comparison leads to choose $k=n-1$ irrespective of the value for $\beta$ and $\mu$.
At this point for $n$ large enough, we can use the Taylor series for the trigonometric functions in \eqref{z_k} and approximate
$|\lambda(n-1)|\approx 1-\frac{2\beta \mu\pi^2}{n^3}
$.
Thus to find the second highest eigenvalue of $\textbf{M}^{c}$ we take:
\begin{equation}
|\lambda_2^{c}|\approx\left(1-\dfrac{2\beta \mu\pi^2}{|\mathcal{V}_c|^3}\right)^{|\mathcal{V}_c|}\approx 1-\dfrac{2\beta \mu\pi^2}{|\mathcal{V}_c|^2}
\end{equation}
\vspace{-.2cm}
\subsection{Proof of Theorem \ref{2_clique_PFS}}\label{proof_PFS_convergence}
\subsubsection*{One shared node}
We will study the evolution of $\bUpsilon_1(t),\bUpsilon_2(t)$, as defined in \eqref{eq:system_state_vector_PFS}.
We will show that there is a unique fixed point irrespective of the initial configuration.
If we have only one node shared between the two cliques we can define, without loss of generality, the two system vectors such that $\pi_{|\Vcal_1|}^1=\pi_{|\Vcal_2|}^2=v'$ with $v'$ being the shared node. 
Except for the update of node $v'$, the updates of the other nodes impact variables only in their clique, i.e. in only one of the system vectors.
We can then consider for $c=1,2$ the following matrix  
\be\label{eq:tilde_M_c}
\tilde{\mathbf{M}}^c\triangleq\prod\limits_{k=1}^{|\Vcal_c|-1}\boldsymbol{M}_{\pi_k^c}=\left[
    \begin{array}{cc}
    \tilde{\mathbf{M}}^c_{(2|\Vcal_c| -1)\times(2|\Vcal_c| -1)}&   \mathbf{0} \\
    \begin{array}{cccc} 0&\cdots&0\end{array} & 1 \\
  \end{array}\right]
\ee 
Note that each $\boldsymbol{M}_{\pi_k^c}$ is a left stochastic matrix and so is the product $\tilde{\mathbf{M}}^c$. Because of the structure of $\tilde{\mathbf{M}}^c$,  $\tilde{\mathbf{M}}^c_{(2|\Vcal_c| -1)\times(2|\Vcal_c| -1)}$ is also left stochastic and it is primitive because it contains all positive elements.
Thus the Perron-Frobenius Theorem ensures that $\tilde{\mathbf{M}}^c_{(2|\Vcal_c| -1)\times(2|\Vcal_c| -1)}$  has exactly one eigenvalue equal to 1 and $2|\Vcal_c|-2$ eigenvalues inside the unit circle. Arguing as in \cite{round_robin_pagliari} we can show that the matrix $\tilde{\mathbf{M}}^c$, has the following two eigenvectors for the eigenvalue equal to one:
\begin{align}
\tilde{\bUpsilon}^\star_{c,1}&=\dfrac{\tilde{\gamma}^{c}}{\tilde{\boldsymbol{D}}^{c}}(\delta,D_{\pi_1^c},\delta,D_{\pi_2^c},\dots,\delta,D_{\pi_{|\Vcal_c|-1}^c},\delta,0)^T \label{eq:tilde_Upsilon_star_c_1}\\
\tilde{\bUpsilon}^{\star}_{c,2}&=\;\;\;\;(0,0,0,0,\dots,0,0,0,1)^T\label{eq:tilde_Upsilon_star_c_2}
\end{align}
with $\tilde{\boldsymbol{D}}^c=\sum_{k=1}^{|\Vcal_c|-1}D_{\pi_k^c}$ and $\tilde{\gamma}^{c}=\dfrac{\tilde{\boldsymbol{D}}^c}{\tilde{\boldsymbol{D}}^c+|\Vcal_c|\delta}$.

The complexity introduced by the presence of the shared node $i$ is that the two nodes 
$\Pre(i,t)$, $\Suc(i,t)$ defined in \eqref{eq:Pre_def}-\eqref{eq:Suc_def} might change over time and belong to different cliques.
We introduce the following two sets
\begin{align}
\overline{\Pre}(i,t)&\triangleq\bigcup_{c\in\mathcal{C}_i}\pre(i,c)\setminus\Pre(i,t)\\
\overline{\Suc}(i,t)&\triangleq\bigcup_{c\in\mathcal{C}_i}\suc(i,c)\setminus\Suc(i,t)
\end{align}
In the case of two cliques, the sets $\overline{\Pre}(v',t),\overline{\Suc}(v',t)$ contain only one node and according to our notation, if for example, $\Pre(v',t)=\pi_{{|\Vcal_1|}-1}^1$ then $\overline{\Pre}(v',t)=\{\pi_{{|\Vcal_2|}-1}^2\}$ and so on.
The update of node $v'$ might cause change in $5$ variables (see Fig.\ref{fig:shared_node_update}) that for short notation we indicate with 
\begin{align}
\Theta_s(t)&\triangleq\Psi_{v'}(t)-\Phi_{\Suc({v'})}(t)\pmod{1}\label{eq:Theta_s_def}\\
\Gamma_{v'}(t)&\triangleq\Phi_{v'}(t)-\Psi_{v'}(t)\pmod{1}\\
\Theta_p(t)&\triangleq\Psi_{\Pre({v'})}(t)-\Phi_{v'}(t)\pmod{1}\label{eq:Theta_p_def}\\
\Theta_{ns}(t)&\triangleq\Psi_{v'}(t)-\Phi_{\overline{\Suc}({v'})}(t)\pmod{1}\\
\Theta_{np}(t)&\triangleq\Psi_{\overline{\Pre}({v'})}(t)-\Phi_{v'}(t)\pmod{1}
\end{align}
\begin{figure}
\includegraphics[scale=0.35]{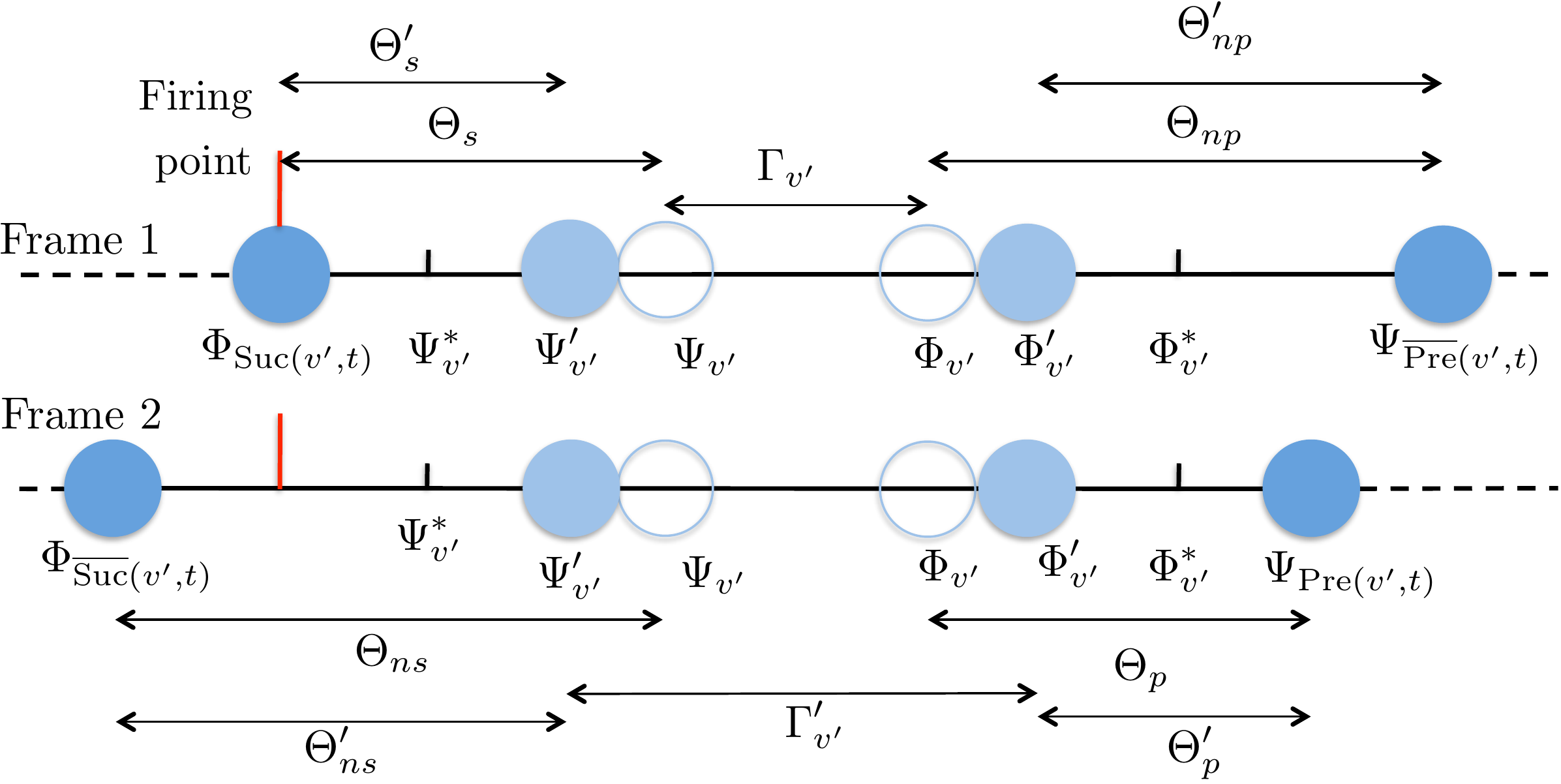}
\caption{Update of the shared nodes and representation of the variables $\Theta_s,\Gamma,\Theta_p,\Theta_{ns},\Theta_{np}$.}\label{fig:shared_node_update}
\end{figure}
Ignoring the dependence on time and using the apex to indicate the updated quantities we can write 
\begin{equation}
\left[\Theta_s',\Gamma',\Theta_p',\Theta_{ns}',\Theta_{np}'\right]^T=\hat{\boldsymbol{M}}_{{v'}}\left[\Theta_s,\Gamma,\Theta_p,\Theta_{ns},\Theta_{np}\right]^T
\end{equation}
where 
\begin{equation}
\hat{\boldsymbol{M}}_{v'}=\left[\begin{array}{ccccc}&&&0&0\\
&\Large{\mathbf{U}_{v'}}&&0&0\\
&&&0&0\\
-\beta\frac{D+\delta}{D+2\delta}&\beta\frac{\delta}{D+2\delta}&\beta\frac{\delta}{D+2\delta}&1&0\\
\beta\frac{\delta}{D+2\delta}&\beta\frac{\delta}{D+2\delta}&-\beta\frac{D+\delta}{D+2\delta}&0&1
\end{array}\right]
\end{equation}
where the first three rows are given by the submatrix $\mathbf{U}_{v'}$ defined in \eqref{eq:U_i_def}
and the last two rows are obtained from 
\begin{align}
\Theta_{ns}'&=\Theta_{ns}-\Theta_s+\Theta_s'\\
\Theta_{np}'&=\Theta_{np}-\Theta_p+\Theta_p'
\end{align}
From the structure of the matrix $\hat{\mathbf{M}}_{v'}$, we can conclude it has $3$ eigenvalues equal to $1$ and $2$ strictly smaller than $1$. In fact, the two right columns are two eigenvectors with eigenvalue $1$ and the top-right block $\mathbf{U}_{v'}$ is positive stochastic, therefore gives a unique eigenvalue equal to $1$ and others $2$ smaller than $1$.
It is then possible to verify the third eigenvector associated with eigenvalue $1$ has to satisfy the following constraint 
\begin{equation}\label{eq:constraint_shared_node}
\Theta_s^*=\Theta_p^*=\frac{\delta}{D_{v'}}\Gamma_{v'}^*
\end{equation}
The fixed point for $\bUpsilon_1(t),\bUpsilon_2(t)$ needs to be such with respect to both the linear map $\bUpsilon_c'=\tilde{\mathbf{M}}_c\bUpsilon_c$ (i.e the updates of all the local nodes) and the update of the shared node just described.
To be eigenvectors in \eqref{eq:tilde_Upsilon_star_c_1}-\eqref{eq:tilde_Upsilon_star_c_2} and meet the constraint $||\bUpsilon_c||_1=1$ we have 
\be 
\bUpsilon_c^\star=(1-\lambda_c)\tilde{\bUpsilon}^\star_{c,1}+\lambda_c\tilde{\bUpsilon}^\star_{c,2}~~\mbox{for}~~c=1,2
\ee 
with $\lambda_c\in (0,1)$.
However, since $\Gamma_{\pi_{|\Vcal_1|}^1}=\Gamma_{\pi_{|\Vcal_2|}^2}=\Gamma_i$ we have $\lambda_1=\lambda_2=\lambda$ and it is possible to verify there is only one value for $\lambda$ that satisfies the constraint in \eqref{eq:constraint_shared_node}, which correspond to the update of the shared node, according to the definition of $\Theta_s(t)$ and $\Theta_p(t)$ in \eqref{eq:Theta_s_def}-\eqref{eq:Theta_p_def}.
This value is given by
$\lambda=\min\limits_{c=1,2}\left\{\dfrac{\gamma^c}{\boldsymbol{D}^c}\right\}D_i$ 
which also tells us that both $\Pre({v'},t)$ and $\Suc({v'},t)$ at the steady state belong to the densest cluster, i.e. the cluster $c'=\argmin\limits_{c=1,2}\left\{\dfrac{\gamma^c}{\boldsymbol{D}^c}\right\}$.
\subsubsection*{Two or more shared nodes}
If we have $|\mathcal{S}_{12}|>1$ shared nodes among two cliques, we need to differentiate between
the case where shared nodes occupy or not occupy consecutive portions of the frame.
For the first case there is a straightforward extension from the previous argument.
Let us define the two system vectors $(\bUpsilon_1(t),\bUpsilon_2(t))$ such that the variables associated to the shared nodes occupy the last positions in both vectors.
Then for $c=1,2$ adjust the definition of $\tilde{\mathbf{M}}^c, \tilde{\boldsymbol{D}}^c, \tilde{\gamma}^c$ making the index $k$ ranging from 1 to $|\Vcal_c|-|\mathcal{S}_{12}|$ instead of $|\Vcal_c|-1$.
With a similar argument as before, we conclude that the fixed points for the two system vectors, taking into account the update of local nodes, lie in the space spanned by: 
\begin{align}\label{basis}
\tilde{\bUpsilon}^\star_{c,1}&=\dfrac{\tilde{\gamma}^{c}}{\tilde{\boldsymbol{D}}^{c}}(\delta,D_{\pi_1^c},\dots,\delta,D_{\pi_{|\Vcal_c|-|\mathcal{S}_{12}|}^c},\delta,0,\dots,0)^T\\
\tilde{\bUpsilon}^\star_{c,j}&=\mathbf{e}_{2(|\Vcal_c|-|\Scal_{12}|)+j}~~ \mbox{for}~~j=2,\dots,2|\Scal_{12}|
\end{align}
If we now consider the consecutive updates of the shared nodes we find the following additional for the fixed point where the node $\Theta_s$ is defined with respect to the last (in order of firing) of the shared nodes and $\Theta_p$ is defined with respect to the first of the shared nodes:
\begin{align}
\Theta_s^{\star}&=\Theta_{\pi_{|\Vcal_c|-|\Scal_c|+2}^c}^{\star}\label{eq:ms2c_constr_1}\\   
\Theta_p^{\star}&=\Theta_{\pi_{|\Vcal_c|}^c}^c\label{eq:ms2c_constr_2}\\
\Gamma_i^{\star}&=\frac{D_i}{\delta}\Theta^*_{\pi_{i+1}^c}~~\mbox{for}~i=|\Vcal_c|-|\Scal_c|+1,\dots,|\Vcal_c|\label{eq:ms2c_constr_3}
\end{align}
The fixed point for each of the system vectors ($c=1,2$) needs to have the following form 
\be 
\bUpsilon_c^\star=\left(1-\sum_{j=2}^{2|\mathcal{S}_{12}|}\lambda_{c,j}\right)\tilde{\bUpsilon}^\star_{c,1}+\sum_{j=2}^{2|\mathcal{S}_{12}|}\lambda_{c,j}\tilde{\bUpsilon}^\star_{c,j}
\ee
and once again we have $\lambda_{cj}=\lambda_j$ for $c=1,2$.
By imposing the constraints in \eqref{eq:ms2c_constr_1}-\eqref{eq:ms2c_constr_2}-\eqref{eq:ms2c_constr_3} we find there is a unique solution that respects the constraints and is consistent with the definition of the predecessor and successor nodes (relative to the firing order), i.e. 
\be
\lambda_j=\begin{cases}&\min\limits_{c=1,2}\left\{\dfrac{\gamma^c}{\boldsymbol{D}^c}\right\}\delta\\&\mbox{for}~~j=(2m+1),m=1,\dots,|\Scal_{12}|-1\\
&\min\limits_{c=1,2}\left\{\dfrac{\gamma^c}{\boldsymbol{D}^c}\right\}D_{\pi_{|\Vcal_c|-|\Scal_{12}|+m}^c}\\&\mbox{for}~~j=2m,m=1,\dots,|\Scal_{12}|
\end{cases}
\ee

For the case 2), the notation is significantly more complicated but the proof can follow the 
same conceptual steps.
Let us start by considering two shared nodes that are not consecutive in the firing order in at least one of the two cliques.
In each clique $c$ we will have two subsets of ${\cal L}_c$, namely ${\cal L}_{c,1}, {\cal L}_{c,2}$ such that the nodes in  each of the two subsets are all consecutive in the firing order.
Then we can define two different matrices 
\be 
\tilde{\mathbf{M}}^{c,j}\triangleq\prod\limits_{k:~\pi_k^c\in\Lcal_{c,j}}\boldsymbol{M}_{\pi_k^c}~~~
\ee
for $j=1,2$.
and also extend the definitions for $\tilde{\boldsymbol{D}}^{c,j},\tilde{\gamma}^{c,j}$, accordingly for $j=1,2$.
Reasoning as before we can see the fixed point $\bUpsilon_c^\star$ needs to mantain the proportionality between the nodes in $\Lcal_{c,1}$ and the ones in $\Lcal_{c,2}$ separately.
Furthermore, the constraint in \eqref{eq:constraint_shared_node} continues to hold and applies to both shared nodes.
Let us assume we are at the fixed point, and there is a certain ``distance'' (portion of the frame) between the two non-consecutive shared nodes, which is assigned in the two cliques to $\Lcal_{c,1}$ and $\Lcal_{c,2}$ respectively.
If we indicate by $i_1$ and $i_2$ the two shared nodes (i.e., the  counterclockwise order of nodes is $i_1,\Lcal_{c,1},i_2,\Lcal_{c,2}$) we find that at the fixed point, i.e., $\forall t>t^*$:
\begin{align}
&\Suc(i_1,t),\Pre(i_2,t)\in\Lcal_{1,1}\leftrightarrow\dfrac{\tilde{\gamma}^{1,1}}{\tilde{\boldsymbol{D}}^{1,1}}<\dfrac{\tilde{\gamma}^{2,1}}{\tilde{\boldsymbol{D}}^{2,1}}\label{eq:two_non_cons_con_1}\\
&\Suc(i_2,t),\Pre(i_1,t)\in\Lcal_{1,2}\leftrightarrow\dfrac{\tilde{\gamma}^{1,2}}{\tilde{\boldsymbol{D}}^{1,2}}<\dfrac{\tilde{\gamma}^{2,2}}{\tilde{\boldsymbol{D}}^{2,2}}\label{eq:two_non_cons_con_2}
\end{align}
In fact, if this is a fixed point, the distance between the two shared nodes is fixed and the other local nodes in each clique share that schedule proportionally, according to their demands.  This implies the guard-spaces in each of the subframes are smaller if the overall demand of that subset of nodes is bigger, giving \eqref{eq:two_non_cons_con_1}-\eqref{eq:two_non_cons_con_2}.
But from this condition, we have a unique fixed point given by the condition in \eqref{eq:constraint_shared_node} that forces the two guard-spaces for any of the shared nodes to be equal and therefore we have enough constraints on the coefficients $\lambda$ to write the two system vectors and find the schedule attainable.
By introducing for $j=1,2$ the cluster 
$c_j^*=\argmin\limits_{c=1,2}\frac{\tilde{\gamma}^{c,j}}{\tilde{\boldsymbol{D}}^{c,j}}$
we have that the scheduling attainable by the unique fixed point is such that ($c,j=1,2$)
\begin{equation}
\Gamma_{i}=\begin{cases}\frac{D_v}{D_{i_1}+D_{i_2}+\frac{\tilde{\boldsymbol{D}}^{c_1^*,1}}{\tilde{\gamma}^{c_1^*,1}}+\frac{\tilde{\boldsymbol{D}}^{c_2^*,2}}{\tilde{\gamma}^{c_2^*,2}}},~&\mbox{for}~~i={i_1,i_2}\\
\frac{\tilde{\gamma}^{c,j}}{\tilde{\boldsymbol{D}}^{c,j}}\frac{\tilde{\boldsymbol{D}}^{c_j^*,j}}{{\tilde{\gamma}^{c_j^*,j}}\left(D_{i_1}+D_{i_2}+\frac{\tilde{\boldsymbol{D}}^{c_1^*,1}}{\tilde{\gamma}^{c_1^*,1}}+\frac{\tilde{\boldsymbol{D}}^{c_2^*,2}}{\tilde{\gamma}^{c_2^*,2}}\right)}D_v~&\mbox{for}~~i\in\Lcal_{c,j}
\end{cases}
\end{equation}
and the guardspaces before the first node in $\Lcal_{c,j}$ ($c,j=1,2$), after the last one and in between are all equal to 
\be 
\delta_{\Lcal_{c,j}}\triangleq\frac{\tilde{\gamma}^{c,j}}{\tilde{\boldsymbol{D}}^{c,j}}\frac{\tilde{\boldsymbol{D}}^{c_j^*,j}}{{\tilde{\gamma}^{c_j^*,j}}\left(D_{i_1}+D_{i_2}+\frac{\tilde{\boldsymbol{D}}^{c_1^*,1}}{\tilde{\gamma}^{c_1^*,1}}+\frac{\tilde{\boldsymbol{D}}^{c_2^*,2}}{\tilde{\gamma}^{c_2^*,2}}\right)}\delta
\ee
For the extension to an arbitrary number of shared nodes $|S_{12}|$ one should consider the shared nodes $i_j$ in $S_{12}$ and the subsets of local nodes $\Lcal_{c,j}$ ($j=1,2,\dots,|S_{12}|$) and apply the same argument to each subset of consecutive nodes to find the unique fixed point.
\vspace{-.2cm}
\subsection{Proof of Proposition \ref{set_multiclique_PFS}}\label{app:proof_set_fixed_points_multiclique}
The proposition is proved by considering a generic subset of consecutive local nodes in any possible clique i.e., a generic $\Lcal_{c,j}$ (see notation introduced in the previous proof, Appendix \ref{proof_PFS_convergence}).
The matrix $\tilde{\mathbf{M}}^{c,j}$ will then have an eigenvector associated with eigenvalue $1$ where the proportional fairness is enforced between the nodes in $\Lcal_{c,j}$ and other several eigenvectors associated with eigenvalue $1$ where the variables $\Gamma_i,~i\in\Lcal_{c,j}$ are equal to 0.
Therefore, all the possible fixed points for the algorithm will respect the {\it partial proportional fairness criterion} as per Definition \ref{def:partial}.
To prove that for more than two cliques we can have in general a set of non-isolated fixed points let us consider the sample topology with $3$ clusters ($c=1,2,3$) and the following properties: $|\Vcal_1|=|\Vcal_3|=4$, $|\Vcal_2|=3$, $|\Scal_{12}|=|\Scal_{23}|=1$, $\Scal_{13}=\emptyset$, $\delta=1, D_v=D=4~~\forall v\in\Vcal$. Then we have that the following configuration (the order of nodes in $\bUpsilon_2^\star$ is $\Scal_{12},\Lcal_2,\Scal_{23}$)
\begin{align*}
\bUpsilon_1^\star&=\bUpsilon_3^\star=\left(\frac{1}{20},\frac{1}{5},\frac{1}{20},\frac{1}{5},\frac{1}{20},\frac{1}{5},\frac{1}{20},\frac{1}{5}\right)^T\\
\bUpsilon_2^\star&=\left(\theta,\frac{1}{5},\frac{1}{10}-\frac{1}{6}\theta,\frac{2}{5}-\frac{2}{3}\theta,\frac{1}{10}-\frac{1}{6}\theta,\frac{1}{5}\right)^T
\end{align*}
is a fixed point for any $\theta\in\left[\frac{1}{20},\frac{3}{10}\right]$, since for all these values $\Pre(i,t)$ and $\Suc(i,t)$ for $i\in\Scal_{13}\cup\Scal_{23}$ continue to remain in $\Vcal_1$ or $\Vcal_3$
and the space left for the only node in $\Lcal_2$ is proportionally distributed between the time schedule assigned to that node and the two guardspaces before and after the schedule.
\vspace{-.2cm}
\subsection{Proof of Proposition \ref{unique_multiclique_PFS}}\label{app:proof_unique_fixed_point_multiclique}
First, recalling the order of the partitions $\mathcal{A}_c$ introduced in \eqref{cluster_sorting}, $\mathcal{A}_1$ is the partition with the highest demand. One can apply Theorem \ref{2_clique_PFS} to each pair of cliques $({\cal V}_1,{\cal V}_c)$ with $1<c\leq |{\cal C}|$ and obtain a possible assignment. In fact, if there are nodes shared among one of these pairs that belong also to other cliques, they have to belong to one of these two partitions by Assumption \ref{ass:demand}. In this case, we can see there is a unique fixed point by having the partitions $\Acal_1,\Acal_c$ assigning the unique schedule obtainable by Theorem \ref{2_clique_PFS} and then apply the argument to each pair of cliques $({\cal V}_2,{\cal V}_c)$ with $2<c\leq |{\cal C}|$. The only additional case we have to consider is that nodes that belong to the pair $({\cal V}_2,{\cal V}_c)$ also belong to ${\cal V}_1$; but these nodes schedules have been already set by the pair $({\cal A}_1,{\cal A}_c)$. For the generic pair case ($\Acal_2,\Acal_c$) Theorem \ref{2_clique_PFS} applies directly if we consider the quantity ${\T}_2$ defined in \eqref{eq:portion_frame_available} in subsection \ref{subsec:fixed_points}. The procedure is then iterated for every pair of cliques until every conflict has been considered, and this proves the Proposition.
\vspace{-.2cm}
\subsection{Proof of Corollary \ref{corollary_line_star}}\label{corollary_PFS}
For a star network, the central node (say $1$) can hear all the rest (i.e. $v=2,\dots,n+1$). There will be $n$ cliques ${\cal V}_c=\{1,c+1\}$ with $c=1,\dots,n$. We simply have to enumerate nodes from $2$ to $n+1$ in decreasing demand order to have ${\cal A}_1=\{1,2\}$ and ${\cal A}_c=\{c+1\}$ for $c=2,\dots,n$. Then Proposition \ref{unique_multiclique_PFS} applies directly and there will be a unique fixed point that respects Definition \ref{def:global} in light of Theorem \ref{2_clique_PFS}.
For a line network, Assumption \ref{ass:demand} is not needed, since every node can belong to no more than two cliques, and then we can just apply Theorem \ref{2_clique_PFS} starting from the two cliques that contain the highest demand node, and then repeat a similar argument as in the proof of Proposition \ref{unique_multiclique_PFS} until we reached the edges of the line.
\end{appendix}
\bibliographystyle{IEEEtran}
\bibliography{Lorenzo_Bibliography}
\begin{IEEEbiography}
[{\includegraphics[width=1in,height=1.25in,clip,keepaspectratio]{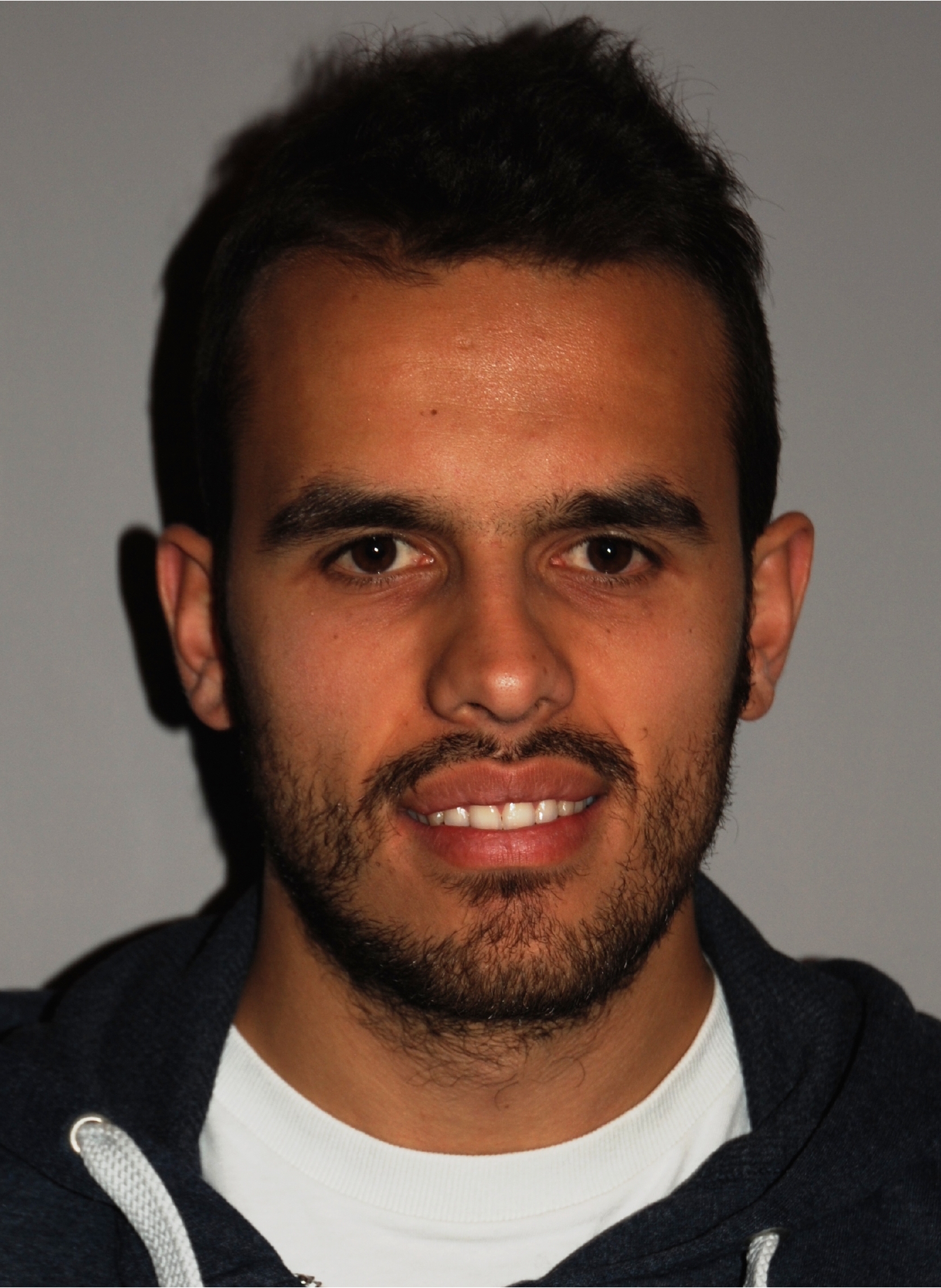}}]
{Lorenzo Ferrari} (S'14) is currently a PhD student in Electrical Engineering at Arizona State University.
Prior to that, he received his B.Sc. and M.Sc degree in Electrical Engineering from University of Modena, Italy in 2012 and 2014 respectively. His research interests lie in the broad area of wireless communications and signal processing. He has received the IEEE SmartGridComm 2014 Best Student Paper Award for the paper ``The Pulse Coupled Phasor Measurement Unit''.
\end{IEEEbiography}
\begin{IEEEbiography}
[{\includegraphics[width=1in,height=1.25in,clip,keepaspectratio]{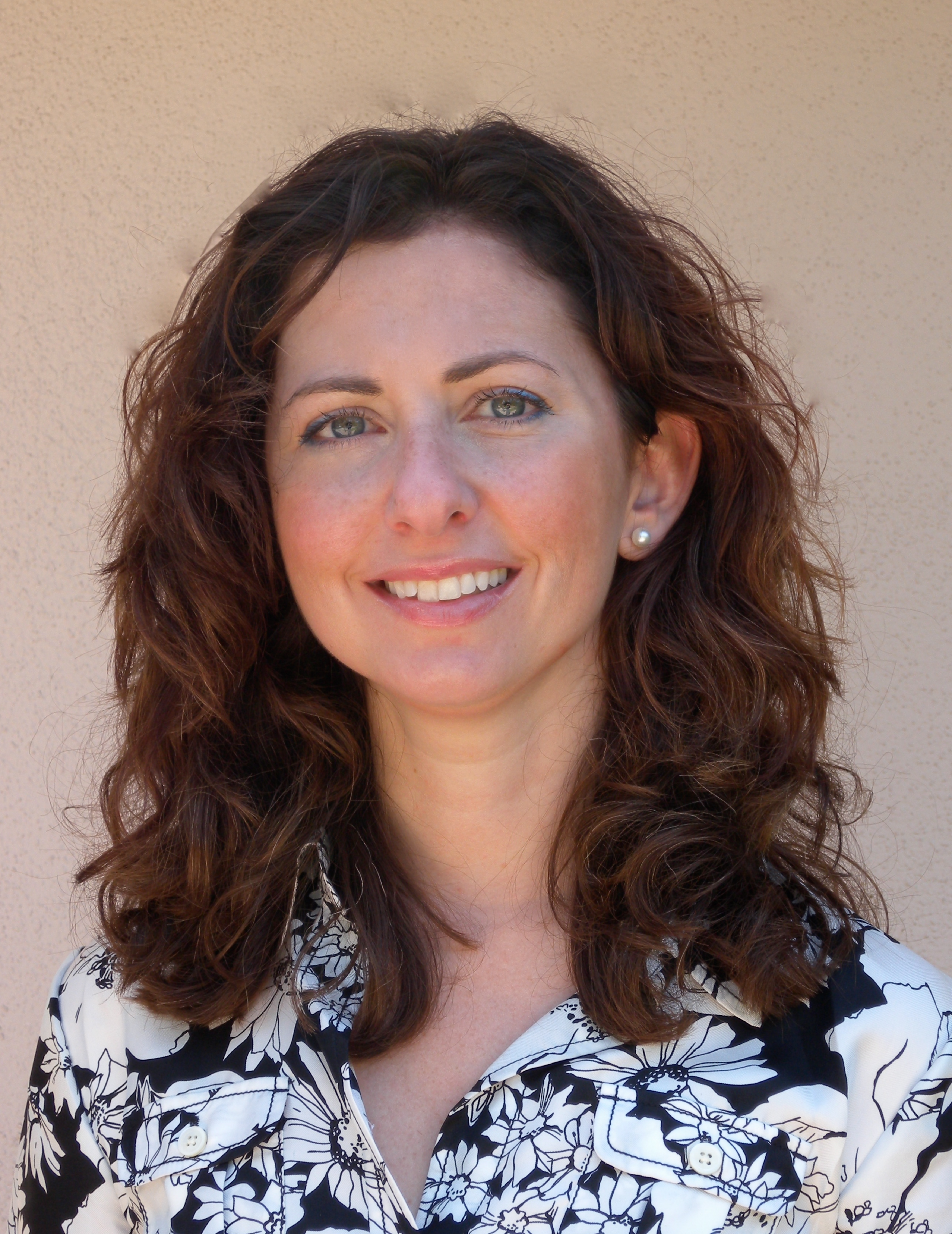}}]
{Anna Scaglione} (F'11) (M.Sc.'95, Ph.D. '99) is currently a professor in electrical and computer engineering at Arizona State University. She was Professor of Electrical Engineering previously at UC Davis (2010-2014), Associate Professor at UC Davis 2008-2010 and at Cornell (2006-2008), and Assistant Professor at Cornell (2001-2006) and at the University of New Mexico (2000-2001).
Dr. Scaglione's expertise is in the broad area of statistical signal processing for communication, electric power systems and networks. Her current research focuses on studying and enabling decentralized learning and signal processing in networks of sensors.
Dr. Scaglione was elected an IEEE fellow in 2011. She served as Associate Editor for the IEEE Transactions on Wireless Communications and on Signal Processing, as EiC of the IEEE Signal Processing letters. She was member of the Signal Processing Society Board of Governors from 2011 to 2014. She received the 2000 IEEE Signal Processing Transactions Best Paper Award and more recently was honored for the 2013, IEEE Donald G. Fink Prize Paper Award for the best review paper in that year in the IEEE publications, her work with her student earned  2013 IEEE Signal Processing Society Young Author Best Paper Award (Lin Li).
\end{IEEEbiography}
\begin{IEEEbiography}
[{\includegraphics[width=1in,height=1.25in,clip]{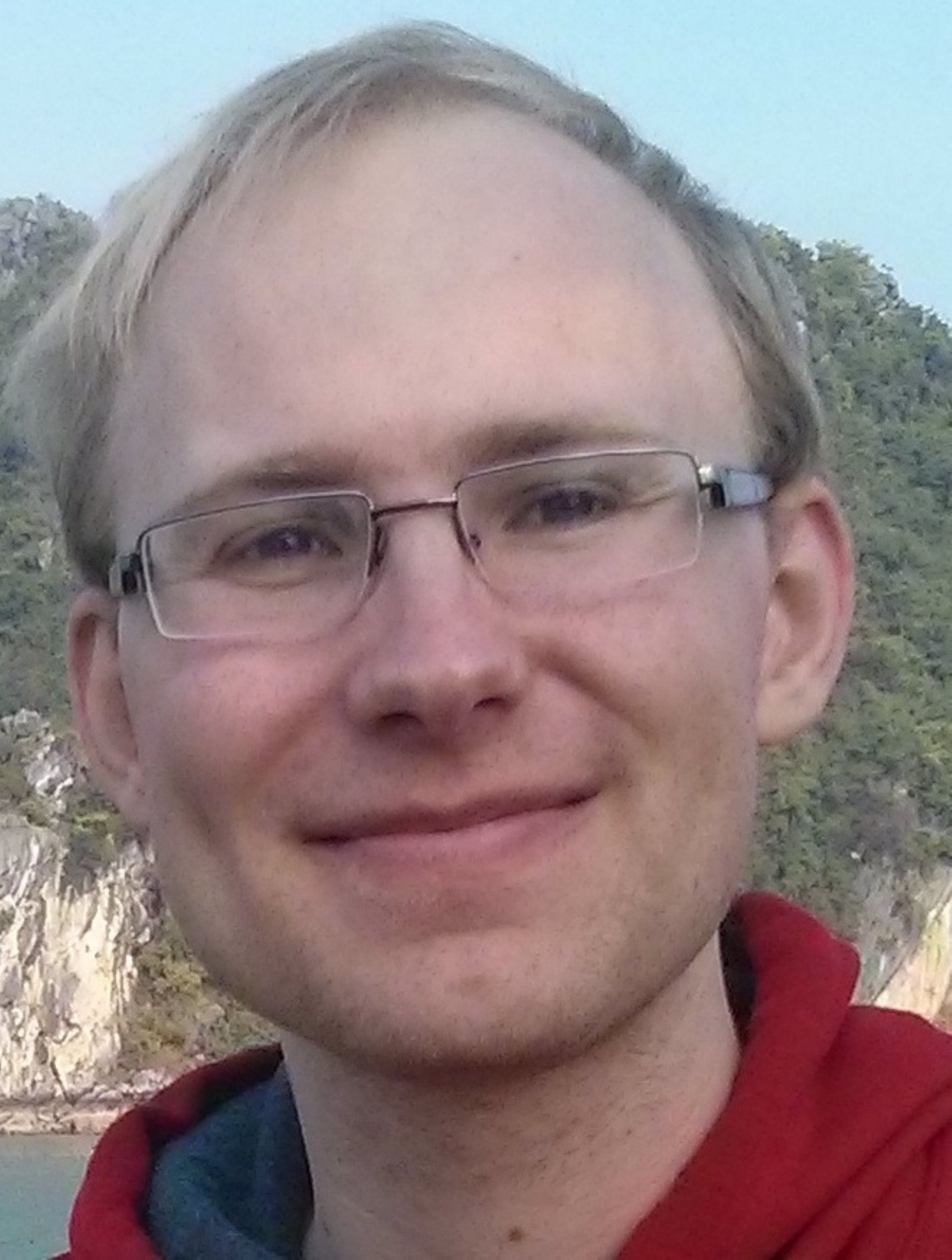}}]
{Reinhard Gentz} (S'15) is currently pursuing a PhD degree in Electrical Engineering at Arizona State University. Prior, he received his B.Sc and M.Sc degree in Electrical Engineering from the Karlsruhe Institute of Technology, Germany in 2010 and 2014 respectively. His research interest are in Embedded Devices, Security, communication technologies and distributed protocols. He has received the IEEE SmartGridComm 2014 Best Student Paper Award for the paper ``The Pulse Coupled Phasor Measurement Unit''.
\end{IEEEbiography}
\begin{IEEEbiography}
[{\includegraphics[width=1in,height=1.25in,clip]{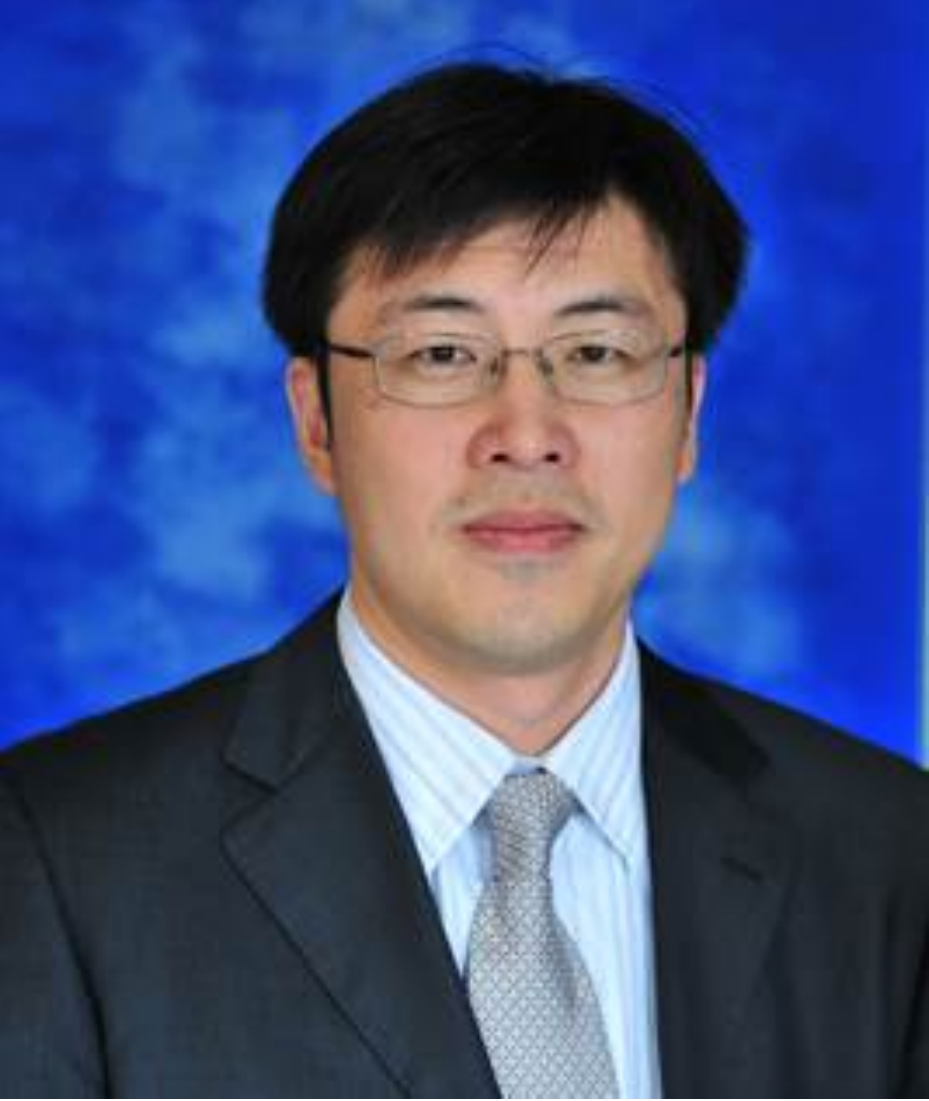}}]
{Y.-W. Peter Hong} (S'01--M'05--SM'13) received his B.S. degree from National Taiwan University, Taipei, Taiwan, in 1999, and his Ph.D. degree from Cornell University, Ithaca, NY, in 2005, both in electrical engineering. He joined the Institute of Communications Engineering and the Department of
Electrical Engineering at National Tsing Hua University, Hsinchu,
Taiwan, in Fall 2005, where he is now a Full Professor. His research interests include physical layer secrecy, cooperative communications, distributed signal processing for sensor networks, and cross-layer designs for wireless networks.
Dr. Hong received the IEEE ComSoc Asia-Pacific Outstanding Young Researcher Award in 2010, the Y. Z. Hsu Scientific Paper Award and the National Science Council (NSC) Wu Ta-You Memorial Award in 2011, and the Chinese Institute of Electrical Engineering (CIEE) Outstanding Young Electrical Engineer Award in 2012. His coauthored paper also received the Best Paper Award from the Asia-Pacific Signal and Information Processing Association Annual Summit and Conference (APSIPA ASC) in 2013. Dr. Hong served as an Associate Editor for IEEE Transactions on Signal Processing (2010-2014) and on Information Forensics and Security (2011-present).
\end{IEEEbiography}
\end{document}